\documentclass[12pt]{article}

\usepackage{comment}
\usepackage[english]{babel}

\usepackage[letterpaper,top=2cm,bottom=2cm,left=3cm,right=3cm,marginparwidth=1.75cm]{geometry}

\usepackage{amsthm}

\usepackage{amsmath,amssymb,setspace}
\usepackage{graphicx}
\usepackage[colorlinks=true, allcolors=blue]{hyperref}
\usepackage{natbib}
\usepackage{econometrics}
\usepackage{multirow}

\newtheorem{theorem}{Theorem}[section]
\newtheorem{kor}[theorem]{Corollary}
\newtheorem{assumption}{Assumption}[section]
\newtheorem{remark}{Remark}[section]

\title{Practically significant differences between conditional distribution functions}
\author{Holger Dette, Kathrin Möllenhoff, Dominik Wied\footnote{Holger Dette: University of Bochum, Kathrin Möllenhoff and Dominik Wied: University of Cologne}}

\begin{document}
\maketitle

\begin{abstract}
In the framework of semiparametric distribution regression, we consider the problem of comparing the conditional distribution functions corresponding to two samples. In contrast to testing for exact equality, we are interested in the (null) hypothesis that the $L^2$ distance between the conditional distribution functions
does not exceed a certain threshold in absolute value. 
The consideration of these hypotheses is motivated by the observation that in applications, it is rare, and perhaps impossible, that a null hypothesis of exact equality
is satisfied and that the real question of interest is to detect a practically significant deviation between the two conditional distribution functions. 

The consideration of a composite null hypothesis makes the testing problem challenging, and in this paper we develop a pivotal test for such hypotheses. Our approach is based on self-normalization and therefore requires neither the estimation of (complicated) variances nor bootstrap approximations. 
We derive the asymptotic limit distribution of the (appropriately normalized) test statistic and show consistency under local alternatives. A simulation study and an application to German SOEP data reveal the usefulness of the method.
\end{abstract}

\newpage

\section{Introduction}
 \def\theequation{1.\arabic{equation}}	
   \setcounter{equation}{0}

In recent years, distribution regression has become a popular approach to model the entire conditional distribution function of an outcome (e.g. income of individuals) given covariates (e.g. the years of education or the years of working experience). 
For many real-world questions in empirical research
this technique provides a very flexible and much more informative inference tool than methods with a direct focus on the conditional mean because of its ability to address aspects such as variability, asymmetry, tail behavior, or multimodality, which play a critical role in contexts involving inequality, risk or discrimination. This allows for richer inference, greater interpretability, and more realistic representations of heterogeneity across populations.
For example, \citet{chernozhukov:2013} used this approach for modeling counterfactual income distributions, \citet{delgado:2022} applied it in duration analysis, \citet{wang:2023} on insurance data, \citet{spady:2025} on the gender wage gap, while \citet{chernozhukov:2020}, \cite{sanchez:2020}, \citet{wied:2024} and \citet{chernozhukov:2025} considered aspects of endogeneity in these models. Indeed, there are various approaches for distribution regression in the literature; a recent survey on this and on possible applications is given by \citet{kneib:2023}.

In this paper, we are concerned with the comparison of two conditional distribution functions in the framework of the semiparametric distribution regression   introduced by \citet{foresiperacchi:1995}. We choose this particular model due to the combination of flexibility and interpretability, noting that our methodology can also be used for other models. The problem of comparing conditional distribution functions has recently been addressed by \citet{hulei:2024} using a different (i.e. the conformal prediction) framework. These authors state two important applications for this test problem: On the one hand, the equality of the training and test distribution in machine learning can be tested. On the other hand, the equality of conditional distributions under different experimental conditions in causal inference can be investigated. Earlier papers about this topic include \citet{zhou:2017}, who proposed smooth tests for the equality of distributions. General specification testing in conditional distribution models was considered by \citet{andrews:1997}, \citet{RotheWied2013} and \citet{troster:2021} among others.

A common feature of all work in this area consists in the fact that  it has  its focus on testing the exact equality 
\begin{align} \label{det92}
H_0^{\rm exact}:~
F^1_{Y|X}  \equiv  F^2_{Y|X}   ~.
\end{align}
of the conditional distribution functions  $F^1_{Y|X}$ and $ F^2_{Y|X}$ 
corresponding to the two samples. In this paper we take a different point of view on the problem of comparing conditional distribution functions. Our work is motivated by the fact that in many cases it is very unlikely that the conditional distribution functions exactly coincide over the full domain of interest.
We thus address a concern  which was already mentioned in  
\cite{berger1987}, who  argued that it  {\it ``is rare, and perhaps impossible, to have a null
hypothesis that can be exactly modeled by a parameter being exactly $0$''} (here the difference  of the conditional distribution functions over the domain of interest). A similar point was raised by  \cite{tukey1991} who mentioned in the context of multiple comparisons of means that  
{\it  ``$\ldots $
All we know about the world teaches us that the
 effects of A and B are always different  - in some 
 decimal place - for any A and B.  Thus asking ``Are
 the effects different?'' is foolish.  $\ldots $''
}  
In the context of comparing two  conditional distribution functions this means that we might not be interested in testing the exact equality as formulated in \eqref{det92}  as we do not believe that this hypothesis can be exactly satisfied (at least not on a microscopic scale). Instead it might be more reasonable to investigate if the deviation between  $ F^1_{Y|X}$ and  $F^2_{Y|X} $ is in some sense ``small''  or not. Therefore we propose to test the hypotheses
\begin{align} \label{det91}
H_0:  d(F^1_{Y|X},  F^2_{Y|X}  ) \leq \Delta  ~~~ {\rm versus}  ~~~
  d(F^1_{Y|X},  F^2_{Y|X}  ) > \Delta  ~,
\end{align}
where $d$ is a metric on the set of distribution functions (we will later specify a particular norm), and $\Delta > 0$ is a given threshold that determines when the deviation between the conditional distribution functions is considered  
as {\it practically significant} {\it or relevant}.  Note that we obtain  the hypotheses in \eqref{det92} from \eqref{det91} for $\Delta =0$, but in the present paper we are not interested in arbitrary small deviations between the conditional distribution functions and therefore we restrict ourselves to  the case $\Delta >0$ throughout this paper.
Thus by testing the hypotheses \eqref{det91} we are trying to answer the question: are the differences between  the two conditional distribution functions   economically significant? In causal inference, one would answer the question if there is a {\it relevant treatment effect}.
\\
 Hypotheses of this form are often called {\it relevant}  hypotheses. The interchanged hypotheses of this type, that is 
$ H_0: d > \Delta  ~\text{ versus }~   H_1:  d \leq \Delta$,  have found considerable attention 
in the biostatistics literature 
for comparing finite dimensional parameters  \citep[see, for example  the monograph by][]{wellek2010testing}, but - despite their importance -   relevant hypotheses  have not been studied intensively in the econometric literature. A few exceptions are the works of see \citet{detteschumann:2024} and 
\cite{kuttadette} who considered these hypotheses in the context of testing the equivalence of pre-trends in
difference-in-differences estimation and 
for validating the homogeneity of the slopes in  high-dimensional panel models, respectively. 

An essential ingredient in implementing any  test for the hypotheses \eqref{det91} is  the specification of the threshold $\Delta$. This choice is case specific,  depends sensitively on the particular problem under consideration and has to be carefully discussed for each application. 
We argue that such a specification is possible in many cases. For example, in the empirical application below, we consider the case of changes in the German income distribution between $2013$ and $2020$.  A  relevant change in the conditional income  distribution  here means that the incomes increased stronger compared to inflation,
which leads to a natural  choice of $\Delta$.  We consider the hypothetical case, in which the model structure remains the same and the incomes increased to exactly the same extent as the consumer prices. If we use a log-linear model, an increase of the incomes of, say, 7\%, implies a shift of the intercept of $0.07$. Then, we calculate the distance between these two conditional distribution functions for some fixed values of the regressors, which yields the value $\Delta$. In order to address model uncertainty at this stage, it would be possible to calculate $\Delta$ for several models and check if the null hypothesis is rejected for the minimum of these values of $\Delta$. 

Moreover,  for applications, where  a specification of $\Delta$ is difficult, we will develop a  method for  determining a threshold from the data, which can serve as measure of evidence for
similarity between the two conditional distribution functions with a controlled type I error $\alpha$ (see Remark \ref{remchoceofdelta} for more details).

In the present paper we will develop a pivotal test for the hypotheses \eqref{det91} where the metric  $d$  is given by the $L^2$-norm. We  consider  the semiparametric distribution regression model,    which has been  proposed  by \citet{foresiperacchi:1995} for
 the conditional distribution functions and   
is introduced in Section \ref{sec2}.  The basic idea consists of using a suitable estimator of the $L^2$-norm of the difference between the conditional distribution functions and reject for large values of this estimate. However,  deriving  critical values either by asymptotic theory or by bootstrap under the composite null hypothesis  
in \eqref{det91} is nontrivial.  In particular the asymptotic distribution depends on nuisance parameters, which are difficult to estimate in practice.  To develop a pivotal test for the hypotheses in \eqref{det91} we propose in Section \ref{sec3}  a self-normalization approach which is based on the weak convergence of a stochastic process estimating the difference $F^{1} _{Y|X}(\cdot |x) - F^{2} _{Y|X}(\cdot |x)   $ of the conditional distribution functions and is used to normalize the difference between the (squared) $L^2$ norm of the estimated difference and the (squared) $L^2$-norm of the true difference of the conditional distribution functions.  As a  
consequence  the limiting distribution is free of nuisance parameters and no bootstrap approximation or variance is required to calculate  quantiles for a decision rule. We  prove that the resulting test  is consistent, has asymptotic level $\alpha$ and can detect local alternatives at a parametric rate. 
Moreover, for applications, where  a specification of the threshold $\Delta$ is difficult, we develop a method for determining a threshold in the hypothesis \eqref{hypotheses} from the data, which can serve as measure of evidence for
similarity between the two conditional distribution functions with a controlled type I error $\alpha$ (see Remark \ref{remchoceofdelta} for more details).
In the same section we also propose an asymptotically pivotal confidence interval for the distance between the two conditional distribution functions and 
discuss extensions to testing for practically significant  (or relevant) endogeneity. In Section \ref{sec4}, we study the finite sample properties of the new test by means of a  simulation study, and  in  Section \ref{sec5}  we illustrate the potential of our approach in an empirical application. Remarks on possible extensions to detect relevant endogeneity and some discussions conclude the paper.

\section{Semiparametric distribution regression}
\label{sec2}
 \def\theequation{2.\arabic{equation}}	
   \setcounter{equation}{0}

We  use  semiparametric distribution regression as  introduced by \citet{foresiperacchi:1995}
to model the dependence between a real valued outcome variable  $Y$ and a $p$-dimensional 
predictor $X$.  Here, the conditional distribution function of $Y$ given the set of regressors $X$ is modeled by 
\begin{align}
    \label{det112}
F_{Y|X}(y|x) = \Lambda(x^\top \beta(y))
\end{align}
for some (known) link function $\Lambda$  (such as the normal distribution) and some function $\beta(y)$.  Although the function $\Lambda$ must be chosen in advance, this model is called semiparametric in the literature. Moreover, in  the finite sample study in Section \ref{sec4}, we demonstrate that our approach is not very sensitive with respect to 
a misspecification of the  link function.

A nice feature of model \eqref{det112} is that  it allows for potentially different parameters for each argument $y$ of the conditional distribution function and that the same time it is easy to interpret because of the linear model structure, which is matched with some link function. Inverting the conditional distribution function leads to models for the conditional quantiles. A potential advantage of this approach compared to standard quantile regression lies in the ability to automatically incorporate jump points. This might be interesting for data sets, where roundings or phenomena such as thresholds for limited-earning jobs are present. Usually, in the literature, no explicit restrictions on $\beta(y)$ such as continuity are imposed. However, it is clear that $\beta(y)$ has to be chosen such that $F_{Y|X}(y|x)$ fulfills the properties of a conditional distribution function if the model is correctly specified.
 Given $x$, this function must be right-continuous,  monotonically increasing and converge to $1$ and $0$ for $y \rightarrow \infty$ and $y \rightarrow -\infty$, respectively. This way, conditional quantiles can be obtained without facing the quantile crossing problem. As pointed out in \cite{wied:2024},
$Y$ may be discretely distributed even if the link function $\Lambda$ is continuous and there is no one-to-one connection between the distribution of $Y$ and the link function.

For an  i.i.d. sample of observations, the conditional distribution function $F_{Y|X}(y|x)$ (more precisely, the corresponding parameter) can be consistently estimated by a Z-estimator similarly as a probit model would be estimated, for example, where the estimation is performed separately for each $y$. To be precise, one introduces the indicator functions $I_y := 1\{Y \leq y\}$ with $E(I_y|X=x) = F_{Y|X}(y|x) = \Lambda(x^\top \beta(y))$. The model can be interpreted as a latent variable model with $I_y = 1$ if $X^\top \beta(y) + U \geq 0$ and $I_y = 0$ otherwise, where  $U$ is  an exogenous random variable (meaning that $X$ is independent of $U$)  with distribution function $\Lambda$.

\section{Relevant differences between conditional distribution functions}
\label{sec3}
 \def\theequation{3.\arabic{equation}}	
   \setcounter{equation}{0}

We will develop a pivotal   test   for a relevant difference between two  conditional distribution functions as formulated in equation \eqref{det91}  in the common  two-sample problem  (for example, comparing  the distributions of  populations in two different years).
For $\ell=1,2$ let $(X_1^{\ell},Y_1^{\ell}),\ldots,(X_{n_{\ell}}^{\ell},Y_{n_{\ell}}^{\ell})$ denote  two  independent samples of i.i.d observations, where $Y_i{^\ell}\in\mathbb{R}$ and   $X_i{^\ell}\in\mathbb{R}^p$. For the conditional distribution function of $Y^{\ell} $ given $X^{\ell} =x$ we  consider the model  
\begin{align}
    \label{det90}
    F^\ell _{Y|X}(y|x) = \Lambda(x^\top \beta_\ell (y)) , 
\end{align}
 where 
$\beta_\ell (y) $ denotes a $p$-dimensional parameter (depending on $y$) and  $\Lambda$ is  a given distribution function.
 For a given threshold  $\Delta > 0$ we are interested in  the relevant hypotheses
\begin{eqnarray}\label{hypotheses}
&~~~~~~~~~~~~~H_0:& || F^1_{Y|X}(\boldsymbol{\cdot}|x) - F^2_{Y|X}(\boldsymbol{\cdot}|x) ||_2 \leq \Delta  ~~~~~~~~~~~~~~~~~~~~~~~~~~~~~~~~~~~~~~~~~~~~~~~~ \\
\nonumber  
\text{  versus  }  & & \\
&~~~~~~~~~~~~~H_1:& || F^1_{Y|X}(\boldsymbol{\cdot}|x) - F^2_{Y|X}(\boldsymbol{\cdot}|x) ||_2 > \Delta ~, ~~~~~~~~~~~~~~~~~~~~~~~~~~~~~~~~~~~~~~~~~~~~~~~~
\label{alternative}
\end{eqnarray}
where $x$ is a fixed covariate, 
\begin{equation*}
||f||_2^2=\int_I f^2(y) dy
\end{equation*}
denotes the (squared) $L^2$-norm of the function $f$ 
and $I\subset \mathbb{R}$ is a given interval of interest for the response $Y$.
The maximum likelihood estimator described in the previous section is denoted by $\hat \beta_\ell (y) $ for each sample $(\ell=1,2)$.
For our approach it is also necessary to consider 
the  corresponding estimators for the samples 
 $(X_1^{\ell},Y_1^{\ell}),\ldots,(X_{\lfloor n_\ell t \rfloor}^{\ell},Y_{\lfloor n_\ell t \rfloor}^{\ell})$ of the first $\lfloor n_\ell t \rfloor$ observations in each sample, where $t\in [0,1] $ is a fixed constant. We denote these  estimators  by $\hat \beta ^{\ell}(t,y)$ and note that  $\hat \beta ^{\ell}(y)= \hat \beta ^{\ell}(1,y)$ ($\ell =1,2$).  Finally, we define
 $$
 \hat F^\ell _{Y|X}(t,y|x) := \Lambda(x^\top \hat \beta^\ell(t,y))
$$
as the sequential estimator of the conditional distribution function $F^\ell_{Y|X}(y|x)$ from the sample  
$(X_1^{\ell},Y_1^{\ell}),\ldots,(X_{\lfloor n_\ell t \rfloor}^{\ell},Y_{\lfloor n_\ell t \rfloor}^{\ell})$ ($\ell =1,2$).
For constructing the test, we then define
the statistics 
\begin{equation}
   \hat T_n= \int_I \hat\Delta^2(1,y|x) dy
   \label{det97}
\end{equation}
and
\begin{equation}
\label{det96}
   \hat V_n^2= \int_\epsilon^1 \Big\{\int_I\hat\Delta^2(t,y|x) dy - \int_I \hat\Delta^2(1,y|x) dy\Big\}^2 dt, 
\end{equation}
where $\epsilon >0 $ is a constant used to  achieve numerical
stability and 
\begin{equation}
\label{det98}
\hat\Delta(t,y|x) := \hat F^1_{Y|X}(t,y|x) - \hat F^2_{Y|X}(t,y|x)
\end{equation}
denotes the difference between the sequential estimates of the conditional distribution functions. Note that (under standard assumptions) 
$\hat\Delta(t,y|x)$ estimates
\begin{align}
\label{det99}
\Delta(y|x) :=  F^1_{Y|X}(y|x) -  F^2_{Y|X}(y|x)
\end{align}
consistently, and consequently the statistic $\hat T_n$ defined in \eqref{det97} is a consistent estimate of 
$$
\int_I\Delta^2 (y|x) dy = \|  F^1_{Y|X}(\cdot |x) -  F^2_{Y|X}(\cdot |x) \|_2^2.
$$
Therefore, we propose to reject the null hypothesis in \eqref{hypotheses}, whenever
\begin{equation}\label{dec_rule}
   \hat T_n > \Delta^2 + q_{1-\alpha} \hat V_n,
\end{equation}
where $q_{1-\alpha}$ is the $(1-\alpha)$-quantile of the distribution of the random variable 
\begin{equation}\label{eq:W}
\mathbb{W}=\frac{\mathbb{B}(1)}{\big[\int_\epsilon^1  \big(\mathbb{B}(t)/t-\mathbb{B}(1)\big)^2dt\big]^{{1}/{2}}}
\end{equation}
and $\mathbb{B}$ is a standard Brownian motion. Note that the quantiles of this distribution  can easily be obtained by simulation. For example, for $\epsilon = 0.1$, the $0.95$-quantile  of the distribution of $\mathbb{W}$ is given by $1.7546$.  We also emphasize that the rejection probabilities of the test  
\eqref{dec_rule} are not very sensitive with respect to  choice of $\epsilon$, as this parameter appears in the quantile and the statistic $\hat V_n$; see Remark \ref{remeps} for more details and Table \ref{tab:level1}  in Section \ref{sec4} for some empirical evidence.

Under the following assumptions, which are similar to the assumptions on the  conditional distribution regression model considered  in \citet{chernozhukov:2013}, we can show that the test \eqref{dec_rule} is consistent and has asymptotic level $\alpha$.
\begin{assumption}\label{ass:1}
{\rm For $\ell=1,2$, the following holds:
\begin{enumerate}
\item The two populations are independent from each other.
\item The conditional distribution function takes the form $F^\ell_{Y|X}(\boldsymbol{\cdot}|x) = \Lambda(x^\top \beta_\ell(y))$, where $\Lambda$ is such that, for fixed $y$, $\beta_\ell(y)$ is consistently estimated by the Z-estimator mentioned in Section 2.
\item $I \subset \mathbb{R}$ is either a compact interval or a finite set. In the former case, the conditional density function $f^\ell_{Y|X}(\boldsymbol{\cdot}|x)$ exists, is uniformly bounded and uniformly continuous.
\item $E||X^\ell||^2 < \infty$ and the minimum eigenvalue of the  $p\times p$-matrix 
\begin{equation*}
E\left(X^\ell X^{\ell^\top} \frac{\lambda(X^\ell \beta_\ell(y))^2}{\Lambda(X^\ell \beta_\ell(y)) [1-\Lambda(X^\ell \beta_\ell(y))]} \right)
\end{equation*}
is bounded away from zero uniformly with respect to  $y \in I $, where $\lambda$ is the derivative of $\Lambda$.
\end{enumerate}
}
\end{assumption}

\begin{theorem}\label{thm4}
    If Assumption \ref{ass:1} holds, and $n_1,n_2\rightarrow\infty$, such that $\tfrac{n_1}{n_1+n_2}\rightarrow c\in(0,1)$, the decision rule \eqref{dec_rule} defines a consistent asymptotic level $\alpha$ test. In particular, 
    \begin{eqnarray*}
  \lim_{ \substack{n_1,n_2\rightarrow\infty \\ {n_1}/({n_1+n_2})\rightarrow c\in(0,1)}}  \mathbb{P}\big(\hat T_n > \Delta^2 + q_{1-\alpha} \hat V_n\big) =
        \begin{cases}
      0, & \text{if}\ \int_I \hat\Delta^2(y|x) dy < \Delta^2 \\
      \alpha, & \text{if}\ \int_I \hat\Delta^2(y|x) dy = \Delta^2  \text{ ~and } \tau^2>0 \\
      1, &  \text{if}\ \int_I \hat\Delta^2(y|x) dy > \Delta^2 ~.
    \end{cases}
    \end{eqnarray*}
\end{theorem}

\begin{proof}
The statement  essentially follows from the weak convergence of the stochastic process
\begin{align} \left\{\sqrt{n}\big(\hat\Delta(t,y|x)- \Delta(y|x)\big)\right\}_{(t,y)\in [\epsilon,1]\times I} \leadsto \left\{\mathbb{H}(t,y)\right\}_{(t,y)\in [\epsilon,1]\times I} ~
    \label{det100}
\end{align}
in the space $\ell^\infty ([\epsilon,1]\times I ) $ of all bounded functions on the set $[\epsilon,1]\times I, $
where $\hat\Delta(t,y|x)$ and $\Delta(y|x)$ are defined in \eqref{det98}  and \eqref{det99}, respectively, and 
$\left\{\mathbb{H}(t,y)\right\}_{(t,y)\in [\epsilon,1]\times I}$ is a two-dimensional centered Gaussian process with 
covariance structure 
$$
{\rm Cov} \big(\mathbb{H}(t_1,y_1),\mathbb{H}(t_2,y_2)\big)= {t_1 \land t_2  \over t_1t_2} \cdot H(y_1,y_2)
$$
for some function $H$ (here we do not reflect the dependence of $\mathbb{H}$ and $H$ on $x$ in the notation).  
This result corresponds to Theorem \ref{thm3}  in the Appendix and is established through several carefully detailed steps provided therein
(for the definition of the function $H$, see equation \eqref{eq:H}).

  Granted with the weak convergence in \eqref{det100} we then proceed as follows. 
We first rewrite the statistic $\hat T_n$  in \eqref{det97} as 
\begin{align}
   \nonumber  
 \hat T_n -  \int_I \Delta^2(y|x) dy &= \int_I \big(\hat\Delta(1,y|x) - \Delta(y|x)\big)^2 dy 
 + 2 \int_I \Delta(y|x)\big(\hat\Delta(1,y|x) - \Delta(y|x)\big) dy \\ 
 &= \mathbb{D}_n(1) +  \int_I \big(\hat\Delta(1,y|x) - \Delta(y|x)\big)^2 dy \nonumber \\
 &= \mathbb{D}_n(1) +  o_{\mathbb{P}} \Big  ( {1 \over \sqrt{n}} \Big ) ~, 
 \label{eq:rep2}
\end{align}
where the process $ \{  \mathbb{D}_n(t)  \}_{t\in [\epsilon,1]}$ is defined by 
\begin{equation*}
    \mathbb{D}_n(t)  = 
    2 \int_I \Delta(y|x)\big(\hat\Delta(t,y|x) - \Delta(y|x)\big) dy 
    .
\end{equation*}
By \eqref{det100} and the continuous mapping theorem,
this  stochastic process  converges weakly in $\ell^\infty [\epsilon,1]$, that is 
\begin{equation}\label{eq:conv_d}
    \big\{\sqrt{n}\ \mathbb{D}_n(t)\big\}_{t\in [\epsilon,1]} \leadsto \left\{ 2 \int_I \Delta(y|x) \mathbb{H}(t,y) dy\right\}_{t\in [\epsilon,1]},
\end{equation}
where  $\mathbb{H}$ is the centered  Gaussian process  in \eqref{det100}. Note that  the estimate \eqref{eq:rep2} is a direct consequence  of  the weak convergence \eqref{eq:conv_d} in $\ell^\infty [\epsilon,1]$.
Obviously, the process on the right hand side of \eqref{eq:conv_d} is a Gaussian process as well, and its covariance structure is given by
\begin{equation}\label{eq:cov}
    {t_1 \land t_2 \over t_1t_2} \tau^2,
\end{equation}
where
\begin{equation}
\label{det115}
    \tau^2 = 4 \int_I\int_I \Delta(y_1|x)\Delta(y_2|x)H(y_1,y_2)dy_1 dy_2
\end{equation}
and $H$ is defined in \eqref{eq:H}.  Consequently it follows from \eqref{eq:conv_d} and \eqref{eq:cov} that
\begin{equation}
\label{eq:conv_dhat}
   \big\{\hat{\mathbb{D}}_n(t)\big\}_{t\in [\epsilon,1]}  \leadsto \left\{ {\tau  \over t} \mathbb{B}(t)\right\}_{t\in [\epsilon,1]}
\end{equation}
 in $\ell^\infty [\epsilon,1]$, 
where $\mathbb{B}$ is a standard Brownian motion.

The statistic $\hat V_n^2$ in \eqref{det96} is used as a self-normalizing term,
and we obtain by similar arguments
the representation
\begin{eqnarray}\label{eq:rep1}
          \hat V_n^2&=& \int_\epsilon^1 \Big\{\int_I\big(\hat\Delta^2(t,y|x) - \Delta^2(y|x)\big)dy -  \int_I \big(\hat\Delta^2(1,y|x)- \Delta^2(y|x)\big)dy\Big\}^2 dt \nonumber\\
          &=& \int_\epsilon^1 \Big\{\int_I\big(\hat\Delta(t,y|x) - \Delta(y|x)\big)^2 dy 
          +2\int_I  \Delta(y|x)\big(\hat\Delta(t,y|x) - \Delta(y|x)\big) dy \nonumber\\
          &-&   \int_I \big(\hat\Delta(1,y|x) -  \Delta(y|x)\big)^2 dy
          -2\int_I \Delta(y|x)\big(\hat\Delta(1,y|x) - \Delta(y|x)\big) dy
          \Big\}^2 dt \nonumber \\
          &=& \int_\epsilon^1 \big \{ \mathbb{D}_n(t) - \mathbb{D}_n(1) \} ^2 dt 
          +  o_{\mathbb{P}} \Big  ( {1 \over \sqrt{n}} \Big ) ~.
\end{eqnarray}
 Combining \eqref{eq:conv_dhat} with the representations \eqref{eq:rep1} and \eqref{eq:rep2} it follows from the continuous mapping theorem that 
\begin{equation}\label{eq:conv_Tn}
   \sqrt{n}\bigg( \hat T_n -  \int_I \Delta^2(y|x) dy, \hat V_n\bigg) \overset{\mathcal D}{\longrightarrow} \tau\left(\mathbb{B}(1),
   \Big[\int_\epsilon^1 \big(\mathbb{B}(t)/t-\mathbb{B}(1) \big)^2 dt \Big]^{\tfrac{1}{2}} \right ).
\end{equation}
Consequently, a further application of the continuous mapping theorem gives (provided that $\tau>0$) 
\begin{equation}\label{eq:conv_Tn2}
   \frac{\hat T_n - \int_I \Delta^2(y|x) dy}{\hat V_n} \overset{\mathcal D}{\longrightarrow}  \mathbb{W},
\end{equation}
where the random variable $\mathbb{W}$ is defined in \eqref{eq:W}.

We will now use this result to prove the statement regarding the rejection probabilities in Theorem \ref{thm4} noting that
a simple calculation gives 
 \begin{eqnarray}
        \mathbb{P}\big(\hat T_n > \Delta^2 + q_{1-\alpha} \hat V_n\big) =  \mathbb{P}\Big (\frac{\hat T_n - \int_I \Delta^2(y|x) dy}{\hat V_n} > q_{1-\alpha} + \frac{\Delta^2 - \int_I \Delta^2(y|x) dy}{\hat V_n}\Big ).~~~~
        \label{det95}
    \end{eqnarray}
    It follows from \eqref{eq:conv_Tn} that $\hat V_n \overset{\mathbb P}{\longrightarrow} 0$, which implies that
    \begin{eqnarray*}
        \frac{\Delta^2 - \int_I \Delta^2(y|x) dy}{\hat V_n} \overset{\mathbb P}{\longrightarrow}
        \begin{cases}
      \infty & \text{if}\ \int_I \hat\Delta^2(y|x) dy < \Delta^2 \\
      0 & \text{if}\ \int_I \hat\Delta^2(y|x) dy = \Delta^2  \\
      -\infty &  \text{if}\ \int_I \hat\Delta^2(y|x) dy > \Delta^2 
    \end{cases}~~.
    \end{eqnarray*}
    Therefore, the assertion of Theorem \ref{thm4} is a consequence of  the weak convergence \eqref{eq:conv_Tn2} and the representation \eqref{det95}.
\end{proof}

The test \eqref{dec_rule}  can detect local alternatives converging to the null hypothesis at a rate $1/\sqrt{n}$. Compared to testing the
exact equality  $F^1_{Y|X}(\boldsymbol{\cdot}|x) = F^2_{Y|X}(\boldsymbol{\cdot}|x)$  of the conditional distribution functions, where local alternatives can be defined in the form $F^1_{Y|X}(\boldsymbol{\cdot}|x) -  F^2_{Y|X} (\boldsymbol{\cdot}|x)= {1 \over \sqrt{n} } h(y|x)$ (for some function $h$), there are several possibilities to define local alternatives  for the hypotheses \eqref{hypotheses} with $\Delta >0 $ such that 
\begin{eqnarray}\label{localhypothese}
 || F^1_{Y|X}(\boldsymbol{\cdot}|x) - F^2_{Y|X}(\boldsymbol{\cdot}|x) ||_2 =  \Delta + \frac{\delta} {\sqrt{n}} + o \big ( \frac{1} {\sqrt{n}} \big )
\end{eqnarray}
for some   constant $\delta  >0$. To be specific we  consider the testing problem  \eqref{hypotheses} and assume a local alternative of the form 
\begin{equation}
\label{localhypothese1}
H_{1,n}:   F^1_{Y|X}(\boldsymbol{\cdot}|x) - F^2_{Y|X}(\boldsymbol{\cdot}|x)  =  g(y|x)   + {\delta \over \sqrt{n}}   h(y|x) 
\end{equation}
where $\delta  >0$ is some   constant  and $ g(\cdot |x)$ and $h(\cdot |x)$ are functions satisfying
$$
\| g(\cdot |x) \|^2 = \Delta^2~,~~\int_I g(y |x) h(y |x) dy = 1
$$ 
(note that a possible choice is $h=g /\Delta^2$, where $g$ is an arbitrary function with $L^2$-norm $\Delta$).
By a straightforward calculation it follows that  in this case \eqref{localhypothese} holds and  the following result shows that the test \eqref{dec_rule} has non-trivial power under these local alternatives. The proof follows by a careful inspection of the proof of Theorem \ref{thm4}  and the arguments given in the online appendix. 

\begin{theorem}\label{thm4loc}
Let  Assumption \ref{ass:1} be satisfied  and assume that  $\tau^2>0 $ and $n_1,n_2\rightarrow\infty$, such that $\tfrac{n_1}{n_1+n_2}\rightarrow c\in(0,1)$.  Under the local alternatives of the form \eqref{localhypothese} the decision rule \eqref{dec_rule} satisfies 
    \begin{eqnarray}
    \nonumber 
  \lim_{ \substack{n_1,n_2\rightarrow\infty \\ {n_1}/({n_1+n_2})\rightarrow c\in(0,1)}}  \mathbb{P}\big(\hat T_n > \Delta^2 + q_{1-\alpha} \hat V_n\big) &=&
      \mathbb{P} \Big ( \mathbb{W} > q_{1-\alpha } 
- {\delta \over \big[\int_\epsilon^1 \big(\mathbb{B}(t)/t-\mathbb{B}(1) \big)^2 dt \big]^{{1}/{2}}}   \Big )   \\
      &>& \mathbb{P} \big ( \mathbb{W} > q_{1-\alpha }  \big ) = \alpha. \label{det93}
    \end{eqnarray}
\end{theorem}

\begin{proof} 
By an inspection of the proof of Theorem \ref{thm3} it is easy to see that the weak convergence in \eqref{det100} also holds under the local alternatives \eqref{localhypothese}. Consequently, the first equality follows by similar arguments as given in the proof of Theorem \ref{thm4}. Moreover,   the random variable 
$\int_\epsilon^1 \big(\mathbb{B}(t)/t-\mathbb{B}(1) \big)^2 dt$
has a Lebesgue density and is therefore positive with probability equal to $1$. The same property implies that the distribution function of the random variable $ \mathbb{W} $
in \eqref{eq:W}
is strictly increasing, which proves  the  strict inequality in the second line of \eqref{det93}.
\end{proof}

\begin{remark} 
\label{remchoceofdelta} 
{\rm 
 Note that the hypotheses  in \eqref{hypotheses}
 and \eqref{alternative} are nested.   Therefore,  it follows  that the rejection of the null hypothesis \eqref{hypotheses}
  by the test \eqref{dec_rule}
  for $\Delta= \Delta_0$ also yields (asymptotically) rejection of $H_{0}$ 
  for all  $\Delta\leq \Delta_0$.
  By the sequential rejection principle, we may simultaneously test the  hypotheses  in \eqref{hypotheses}  and \eqref{alternative}  for different $\Delta \geq 0$ 
  starting at $\Delta  = 0$ and 
   increasing  $\Delta $ to 
   find the minimum value of $\Delta $, say
  \begin{align}
     \label{det113}
 \hat \Delta_\alpha:= \min \Big \{ \big \{ 0  \} \cup  
 \big \{\Delta \ge 0 \,| \,  \hat T_n > \Delta^2 + q_{1-\alpha} \hat V_n  \big  \} \Big \} = \big \{   (
 T_n -  q_{1-\alpha} \hat V_n  )
\lor 0  \big \}^{1/2} 
 \end{align}
   for which 
   $H_0$  in \eqref{hypotheses} is not rejected.
  The quantity  $\hat \Delta_\alpha $ could be interpreted as a measure of evidence against the null hypothesis in \eqref{hypotheses}   (note that  the null hypothesis
  is accepted  for all  thresholds 
  $ \Delta \geq  \hat \Delta_\alpha $ and rejected for 
  $ \Delta <   \hat \Delta_\alpha $).
  This means that higher values of $\hat \Delta_\alpha$ yield stronger evidence against the null hypothesis of a small difference.
  In this sense, the question of a reasonable
choice of the threshold $\Delta$ to test the relevant hypotheses may be postponed until after seeing the data. 

Moreover, if one is not directly interested in testing, a pivotal confidence interval for the $L^2$-distance between the conditional distributions can be constructed. To be precise, note that it follows from   the proof of Theorem \ref{thm4}  
that
$$
\lim_{n \to \infty  } \mathbb{P}
\Big ( \frac{\hat T_n - \int_I \Delta^2(y|x) dy}{\hat V_n}  \leq  q_{1- \alpha}  
\Big )  =  1- \alpha~, 
$$
and  therefore an 
asymptotic one-sided $(1-\alpha)$-confidence interval for  the deviation $\| \Delta^2(\cdot |x) \|_2 $ between the two conditional distribution functions is given by 
\begin{align*}
   [\hat \Delta_\alpha,\infty)=   \Big [ \big \{   (
 \hat T_n -  q_{1-\alpha} \hat V_n  )
\lor 0  \big \}^{1/2}  ,\infty    \Big )~.
\end{align*}
}
\end{remark}

\begin{remark} 
\label{remeps}
{\rm 
 Note that the test \eqref{dec_rule} depends on the specification of the constant $\epsilon> 0$,  which has been introduced in the definition of the statistic $\hat V_n$ to achieve numerical stability. We give a heuristic argument that the test is not very sensitive with respect to this choice.  To see this, we define the random variable 
\begin{equation}
\label{hol51}
\mathbb{V}_\epsilon =\Big(
\int_\epsilon^1  \big(\mathbb{B}(t)/t-\mathbb{B}(1)\big)^2  dt
\Big)^{1/2} ,
\end{equation}
corresponding to the denominator in \eqref{eq:W},
$M^2 = \int_I  \Delta^2 (y|x) dy $ as the  squared $L^2$-norm of the difference of the conditional distribution functions and make the dependence of the quantile $q_{1-\alpha}$ on the distribution of $\mathbb{W} = \mathbb{B}(1)/ \mathbb{V}_\epsilon  $ more explicit using the notation $q_{1-\alpha} (\mathbb{W}) $.
With these notations we obtain for the probability of rejection
\begin{align}  
\mathbb{P}\big( \hat T_n > \Delta^2 + q_{1-\alpha} (\mathbb{W} )  \hat V_n \big) 
& \approx 
 \mathbb{P}\Big( \mathbb{B}(1) > \frac{\sqrt{n} 
  (\Delta - M^2)}{\tau}  + \mathbb{V }_\epsilon \cdot q_{1-\alpha}(\mathbb{B}(1)/ \mathbb{V}_\epsilon )   \Big)  ,  \label{h40a}
\end{align}
where we have used the weak convergence in \eqref{eq:conv_Tn} for the approximation of the probabilities.
Observe that in the last expression  only the quantity $\mathbb{V }_\epsilon \cdot q_{1-\alpha}( \mathbb{B}(1)/\mathbb{V}_\epsilon)$ depends 
on the constant $\epsilon$, which enters in the definition of the random variable $\mathbb{V}_\epsilon$.
However,  for fixed $v>0 $ we have $ {v }\cdot q_{1-\alpha}( \mathbb{B}(1)/v) = q_{1-\alpha}( \mathbb{B}(1) )$,
which gives a  heuristic explanation  why 
the 
probability in  \eqref{h40a} is not very sensitive with respect to the choice of $\epsilon$ (which was also observed empirically, see Table \ref{tab:level1}  in Section \ref{sec4}). 
Note that the same argument applies if the Lebesgue measure in the definition of  the statistic $\hat V_n$ is replaced by a different measure. 
}
\end{remark}

\begin{remark} ~~
\label{remtimeseries} 

{\rm 
   (a) It follows from its proof that  the  statement of Theorem \ref{thm4} is also valid under less restrictive assumptions, for example  under time series assumptions. More precisely, Theorem \ref{thm4} remains true and the decision rule \eqref{dec_rule} defines a pivotal, consistent and  asymptotic level $\alpha$-test for the hypotheses \eqref{hypotheses}, whenever the weak convergence in \eqref{det100} can be established.  Now a careful inspection of the proofs in the Appendix (in particular the proof of Theorem \ref{thm1} and \ref{thm2}) shows that such results can be obtained under appropriate mixing  \citep{Bradley.2007}, physical dependence \citep{Wu2005} or $m$-approximability \citep{HrmannKok} conditions (however a proof of a statement corresponding to Theorem \ref{thm2} would change substantially).
A similar comment can be made regarding the independence assumption on the two samples.
\medskip

(b) The limit distribution of the statistic in \eqref{eq:conv_Tn} is not pivotal  if $\Delta (y|x) \equiv 0 $.  To see this, note that it follows by  the weak convergence \eqref{det100}, the definition of $T_n$ and $\hat V_n$ in \eqref{det97} and \eqref{det96}, respectively, and the continuous mapping theorem  that 
\begin{align}
    \label{det89}
 \frac{\hat T_n }{\hat V_n} \overset{\mathcal D}{\longrightarrow} { \int_I \mathbb{H}^2(1,y) dy \over
 \big [ 
 \int_\epsilon^1 \big\{\int_I\mathbb{H}^2(t,y) dy - \int_I \mathbb{H}^2(1,y) dy\big\}^2 dt \big]^{1/2}} 
\end{align}
where $\mathbb{H}$ denotes the Gaussian process on  $[\epsilon,1]\times I $  with covariance structure defined in equation \eqref{det3} of the Appendix. Consequently, the distribution of the right hand side of \eqref{det89} is not pivotal and  its quantiles cannot be used for a decision rule which rejects 
the null hypothesis \eqref{hypotheses} with $\Delta=0$  for large values of ${\hat T_n }$.
\medskip

(c) 
Note that in the formulation of the hypothesis pair, we have fixed the value of the predictor $x$. This is particularly relevant for our empirical application, where we compare the tests' results for different values of $x$. The approach could easily be extended to an aggregation over the regressors.
}
\end{remark}

\begin{remark}
{\rm The test approach can in principle also be used for detecting ``relevant endogeneity'', compare \citet{chernozhukov:2020}, \cite{sanchez:2020}, \citet{wied:2024} and \citet{chernozhukov:2025} for results about endogeneity in distribution regression models. \citet{wied:2024} considers another conditional distribution function, in which potential endogeneity of one regressor $Y_2$ is addressed and an instrument $Z$ is available. This function is defined as 
$$
F^V_{Y|X,Y_2}(y|x,y_2) = \int P(I_y^* \geq 0|X=x,Y_2=y_2,V=v) d F_V(v)
$$
with
\begin{eqnarray*}
I_y^* &=& X'\beta_1(y) + Y_2 \beta_2(y) + U(y) \\
Y_2 &=& X'\gamma_1 + Z'\gamma_2 + V.
\end{eqnarray*}
Here, $U(y)$ and $V$ are latent variables with mean zero, for which there exists a decomposition 
\begin{equation*}\label{decomposition}
U(y) = \alpha_1 V + \alpha_2 \epsilon(y).
\end{equation*}
In this one-sample setting, one could test for relevant endogeneity with the hypothesis pair
\begin{eqnarray*}
&H_0:& || F_{Y|X,Y_2}(y|x,y_2) - F^V_{Y|X,Y_2}(y|x,y_2) || \leq \Delta \\
&H_1:& || F_{Y|X,Y_2}(y|x,y_2) - F^V_{Y|X,Y_2}(y|x,y_2) || > \Delta.
\end{eqnarray*}
The structure of the test statistic would be the same, using appropriate estimators of $F_{Y|X,Y_2}(y|x,y_2)$ and $F^V_{Y|X,Y_2}(y|x,y_2)$. While an  estimator for the first conditional distribution function  is obtained similarly as described in Section \ref{sec2}, the estimation of the function $F^V_{Y|X,Y_2}(y|x,y_2)$ requires the use of a control function and an additional average step, as described in \citet{wied:2024}.
}
\end{remark}

\section{Monte Carlo simulations}\label{sec4}
 \def\theequation{4.\arabic{equation}}	
   \setcounter{equation}{0}

In this section, we demonstrate the application of the proposed testing procedure \eqref{dec_rule} and present finite sample evidence. Throughout this section, we assume the same sample size for both groups, i.e. $n=n_1=n_2$. We consider two different scenarios, where the first one is given by a simple linear model and the second one is inspired by the application discussed in Section \ref{sec5}. For both scenarios, we will consider different sample sizes and vary the underlying distributions. In addition, we will use one scenario  to investigate the robustness of the procedure with respect to misspecifcation of the link function $\Lambda$.
Further, in order to investigate the effect of the parameter $\epsilon$ in definition of the self-normalizing statistic \eqref{det96}, we will simulate the type I error rates of the test for three different choices of  $\epsilon \in \{0.05,0.1,0.2\}$, in all scenarios under consideration.
All simulations have been run using R Version 4.4.3 with a  total number of $1000$ simulation runs  for each configuration, where  the significance level was chosen as $\alpha=0.05$. 

In Scenario 1 we assume that
\begin{equation}\label{sim:scen1}
    Y_1=1+X_{12}+U_1,\ Y_2=1.7+X_{22}+U_2,
\end{equation}
where for both groups $\ell=1,2$, the error term $U_\ell\sim \mathcal{N}(0,\sigma^2)$ is independent of the regressor $X_{\ell 2}\sim \mathcal{N}(0,1)$, such that $F^\ell_{Y|X}(y|x) = \Lambda(x^\top \beta_\ell(y))$, where $x^\top=(1,x_{12}) $. The  link function $\Lambda$ is the distribution function of a normal distribution with mean $0$ and variance $\sigma^2$ such that   $\beta_1(y)=(y-1,-1)$, $\beta_2(y)=(y-1.7,-1)$.
For a given $x^\top =(1,1)$ we obtain, depending on the choice of $\sigma^2$, the following underlying true values 
\begin{equation}\label{sim:scen1_values}
|| F^1_{Y|X}(\cdot  |(1,1)) - F^2_{Y|X}(\cdot |(1,1)) ||^2_2= 
\begin{cases}
    0.256 \text{  for $\sigma^2=0.25$}\\
    0.135 \text{  for $\sigma^2=1$}\\
    0.066 \text{  for $\sigma^2=4$}
\end{cases}
\end{equation}
for $y \in I=[-1,5]$.

Figure \ref{fig:scen1} displays the empirical rejection probabilities of the test \eqref{dec_rule} for the three choices of the variance $\sigma^2$ and three different sample sizes $n=200,500,1000$ in dependence of the threshold $\Delta^2$, where the parameter $\epsilon $ in the self-normalizing statistic is chosen as $\epsilon =0.1$. 
For all three sub-figures, the red vertical dashed line indicates the true underlying value of the squared $L^2$-distance between the conditional distribution functions given in \eqref{sim:scen1_values} and thus the margin of the null hypothesis in \eqref{hypotheses}. Therefore, the empirical rejection probabilities on the right of this line correspond to the situation under the null hypothesis and thus display simulated type I error rates, whereas the values on the left of this line correspond to the alternative and hence represent the simulated power.  
We observe that all simulated type I error rates are close to or well below the chosen significance level of $\alpha=0.05$, approaching $0$ for larger choices of the threshold $\Delta^2$.

Next we investigate the sensitivity of the test \eqref{dec_rule} with respect to the choice of the parameter $\epsilon$ in the self-normalizing statistic 
\eqref{det96}. We concentrate on the margin 
$|| F^1_{Y|X}(\boldsymbol{\cdot}|x) - F^2_{Y|X}(\boldsymbol{\cdot}|x) ||_2 =  \Delta $ of the hypothesis \eqref{hypotheses}
and display in 
 Table \ref{tab:level1} the simulated type I error of the test \eqref{dec_rule} for  three different choices of  $\epsilon \in \{0.05,0.1,0.2\}$. We observe that the choice of $\epsilon$ does not influence the performance of the test, as the results are qualitatively the same across all three configurations under consideration.
 
Finally, considering the situation under the alternative, we note that the simulated power approaches $1$ if the variance under consideration is sufficiently small or the sample size large. These findings underline the theoretical properties of the test stated in Theorem \ref{thm4}.

\begin{figure}[h]
\begin{center}
\includegraphics[width=0.32\textwidth]{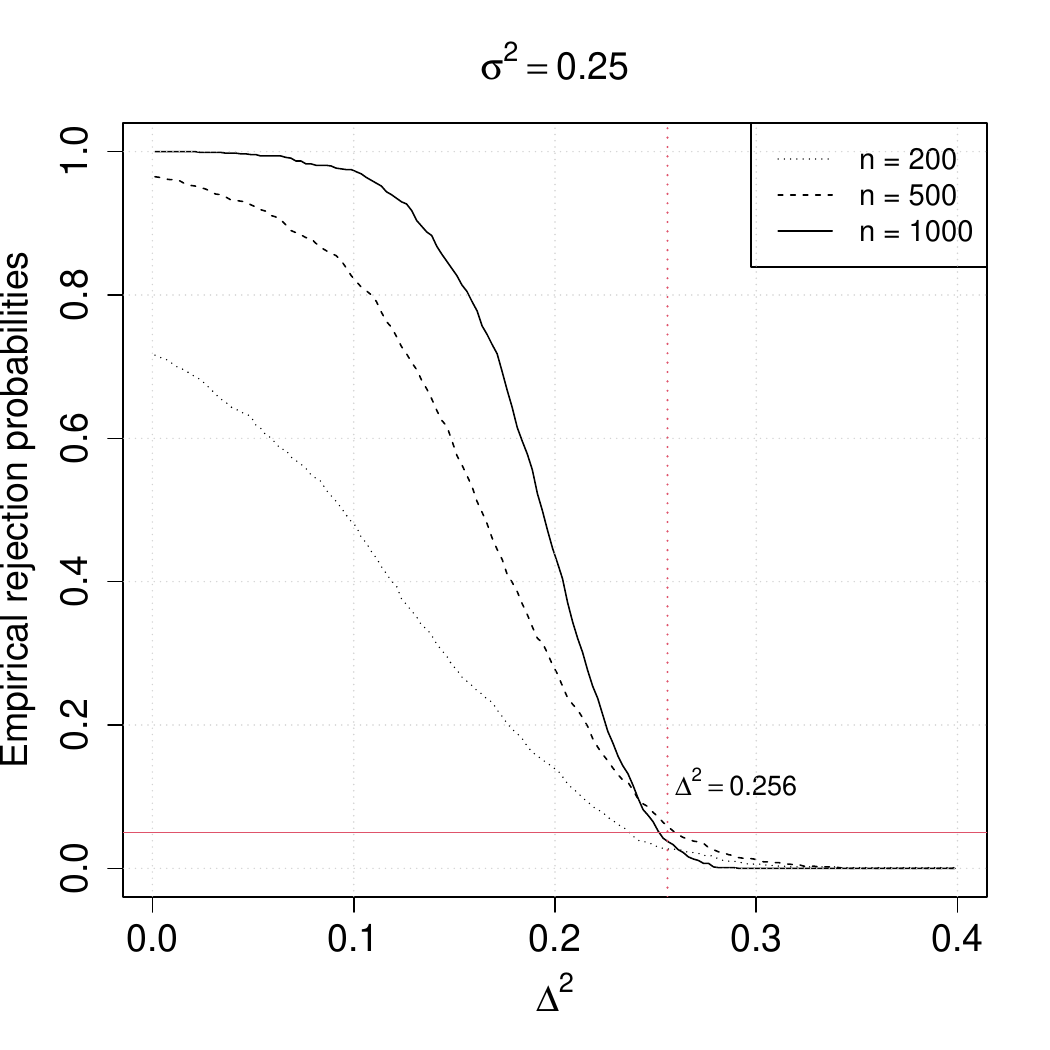}
\includegraphics[width=0.32\textwidth]{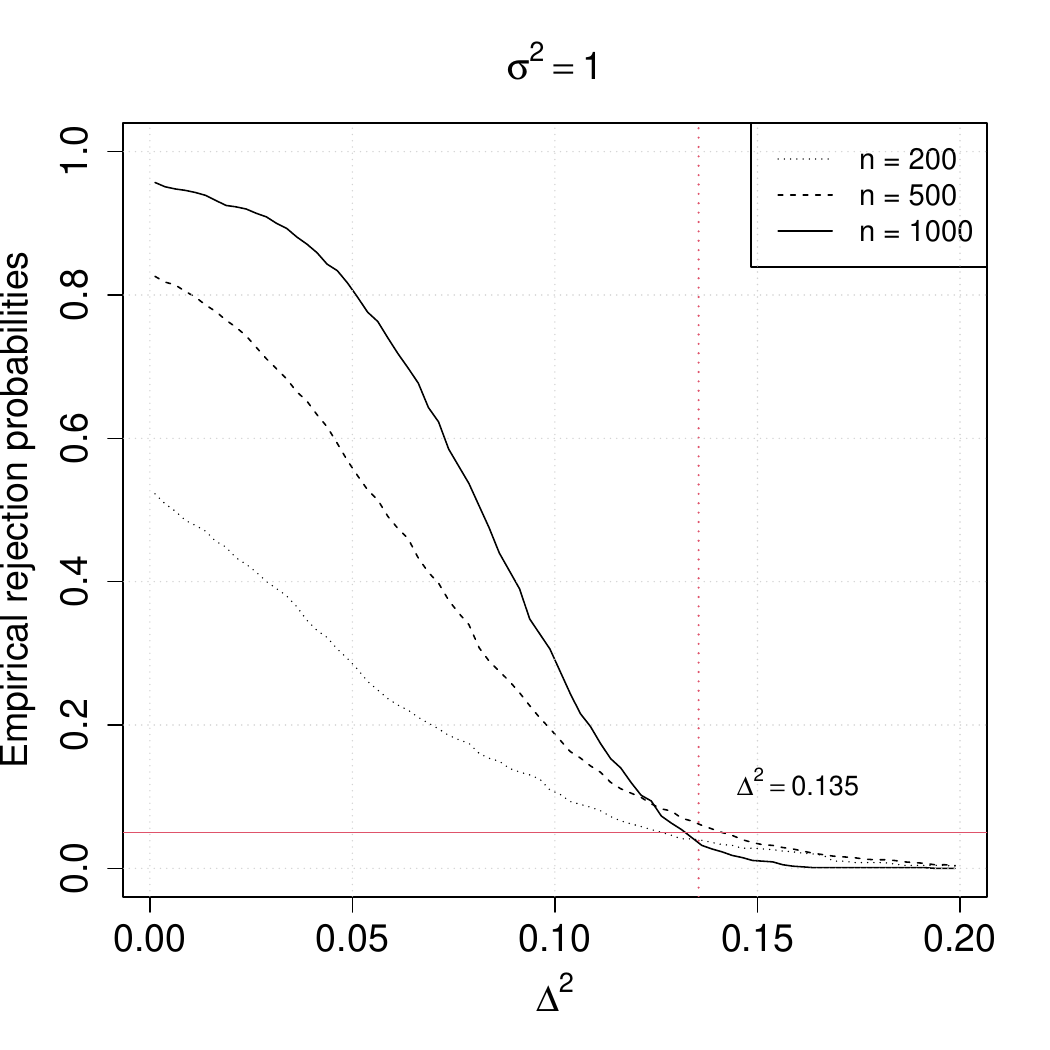}
\includegraphics[width=0.32\textwidth]{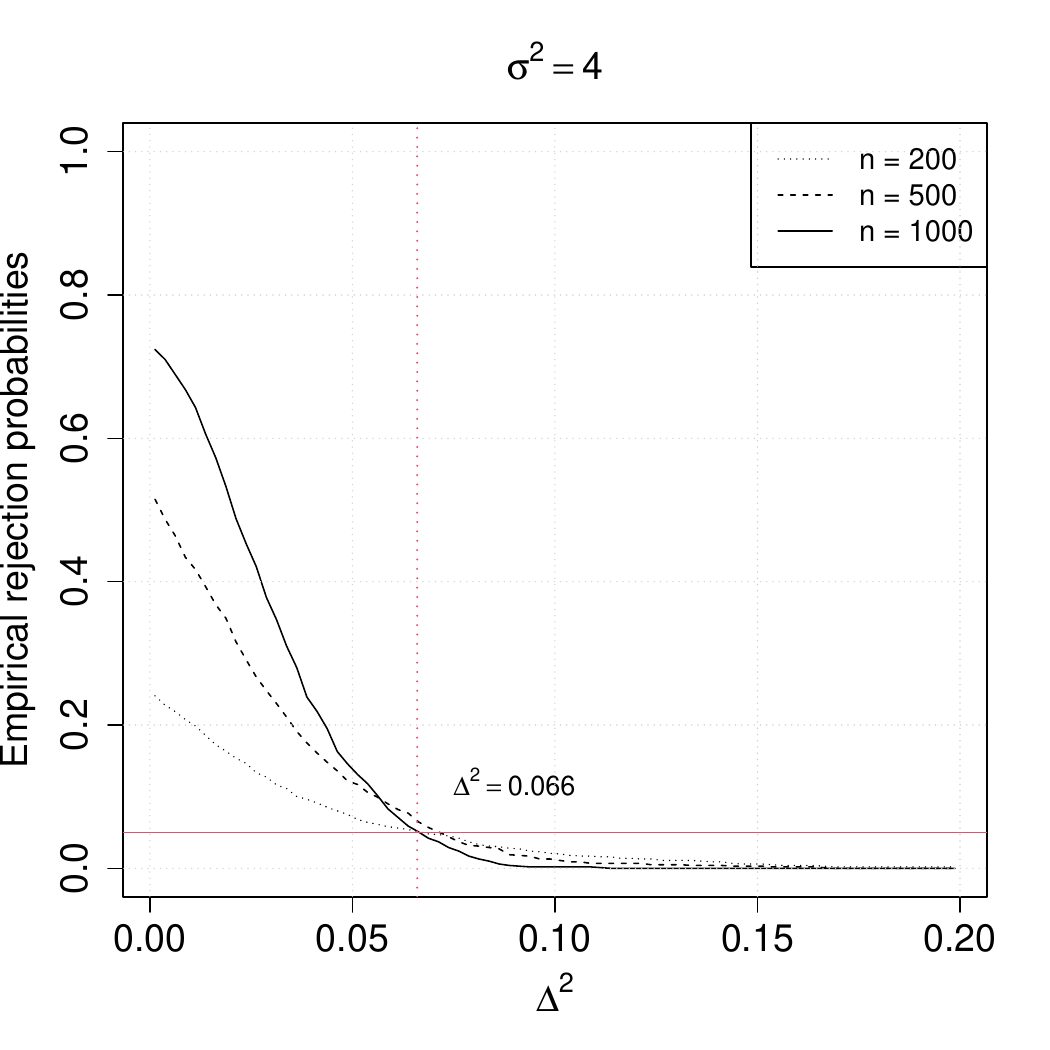}
\caption{Empirical rejection probabilities of the test \eqref{dec_rule} for $\epsilon=0.1$, three choices of $\sigma^2$ and different sample sizes  in dependence of the threshold $\Delta^2$. The red solid line indicates the significance level of $\alpha=0.05$, the red dashed line indicates the true underlying value and thus the margin of the null hypothesis in \eqref{hypotheses}.}
\label{fig:scen1}
\end{center}
\end{figure}

We continue investigating  the robustness of the test procedure. For this purpose  we  assume that the true underlying link function is given by the distribution function of a  logistic distribution with location parameter $0$ and scale parameter $1$, whereas estimation is performed under the assumption of a standard normal distribution. For this scenario the true underlying value  of the squared $L^2$-distance between the conditional distribution functions is given by  
$$\| F^1_{Y|X}(\cdot |(1,1)) - F^2_{Y|X}(\cdot | (1,1)) ||^2_2= 0.080,
$$ 
where $F^1$ and $F^2$ now denote the distribution functions of the standard logistic distribution. The functions $F^1_{Y|X}(\cdot |(1,1)) $ and 
$F^2_{Y|X}(\cdot |(1,1))$ are shown in the left panel of Figure \ref{fig:scen1_robustness}, while the right panel shows the simulated rejection probabilities of the test \eqref{dec_rule}. We observe that even in this situation of misspecification, the test keeps its nominal level for all three sample sizes under consideration while still achieving reasonable power, approaching a maximum of $0.8$ for the largest sample size of $n=1000$.

\begin{figure}[h]
\begin{center}
\includegraphics[width=0.49\textwidth]{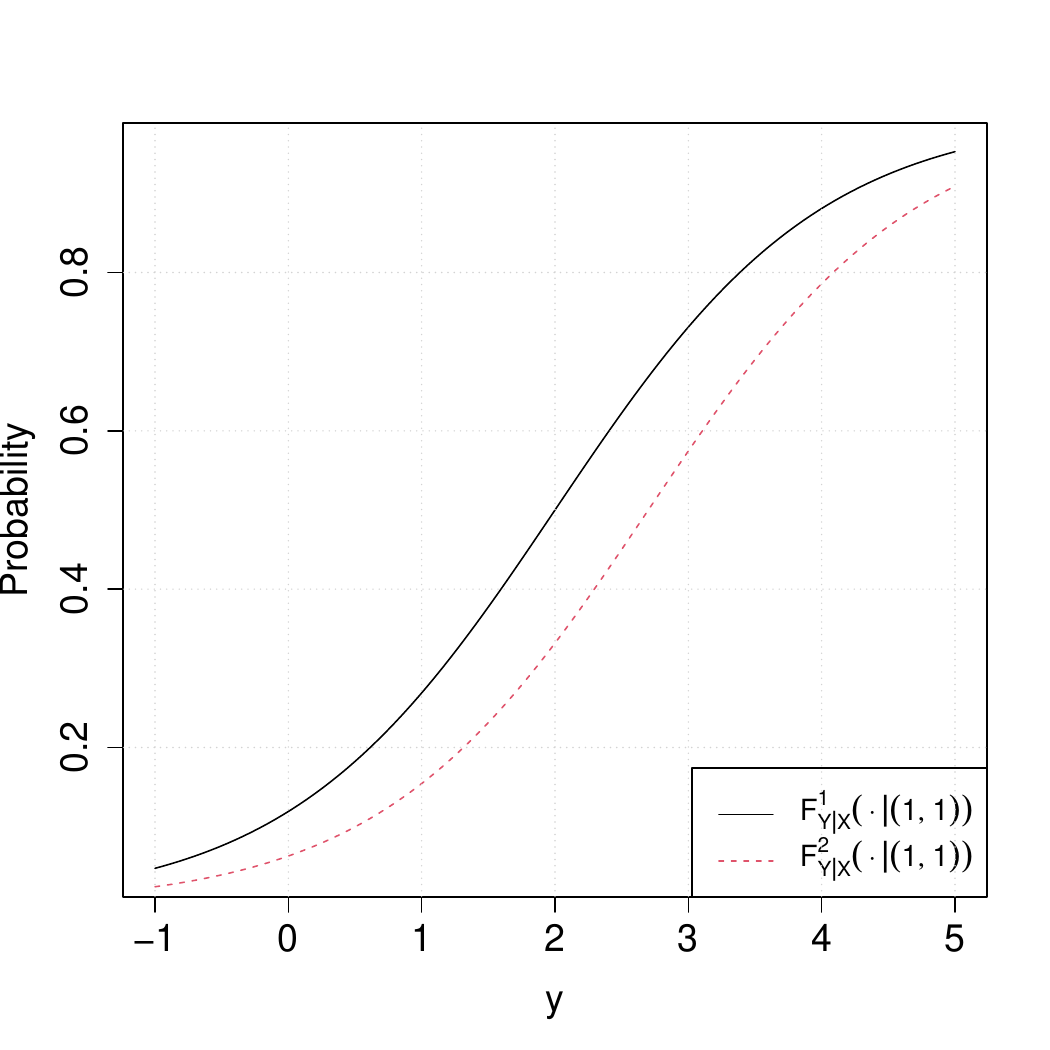}
\includegraphics[width=0.49\textwidth]{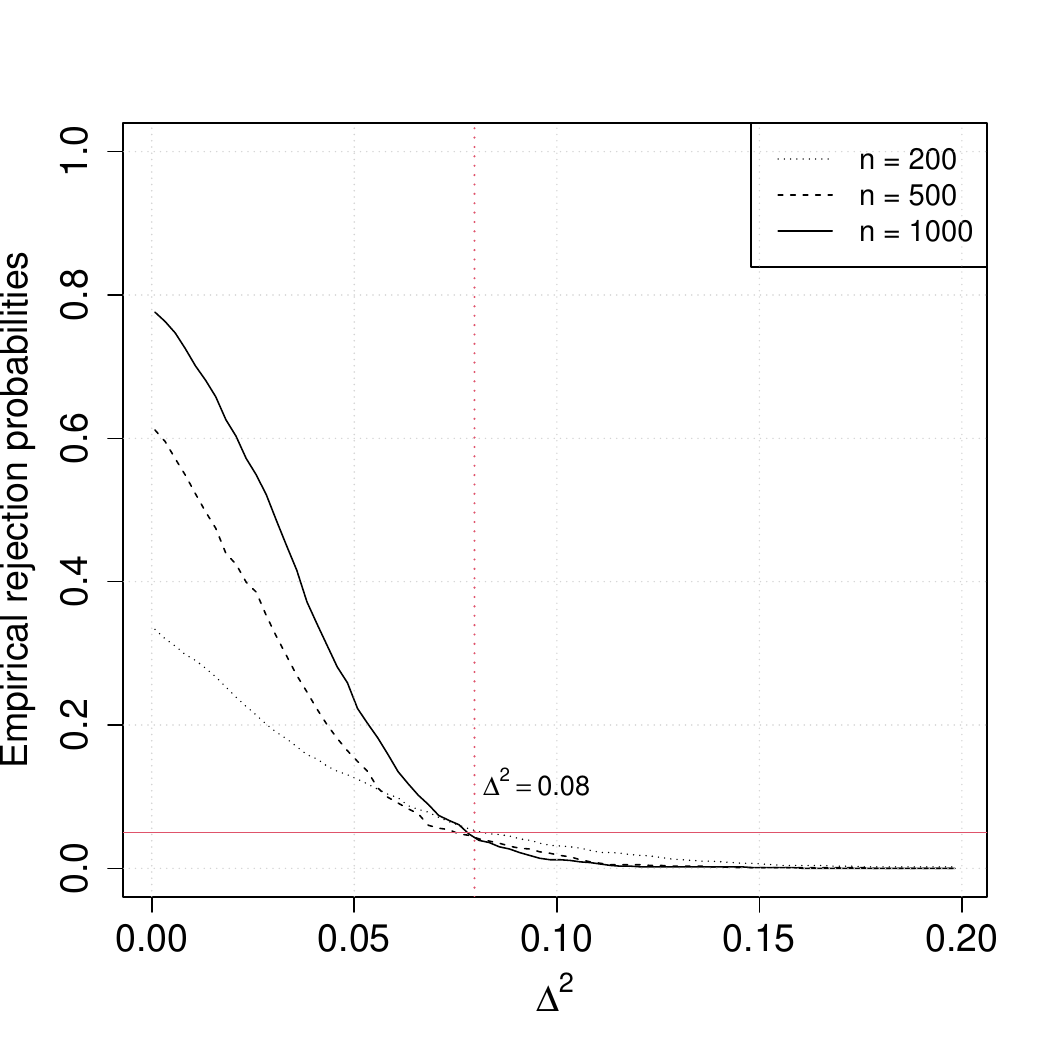}
\caption{Left panel: The true underlying conditional distribution functions
$F^1_{Y|X}(\cdot |(1,1)) $ and 
$F^2_{Y|X}(\cdot |(1,1))$ (link function is a logistic distribution). Right panel: Empirical rejection probabilities of the test \eqref{dec_rule} for  different sample sizes  in dependence of the threshold $\Delta^2$, where estimation is performed under the assumption of
a standard normal distribution as  link function. The red solid line indicates the significance level of $\alpha=0.05$, the red dashed line indicates the true underlying value and thus the margin of the null hypothesis in \eqref{hypotheses}.}
\label{fig:scen1_robustness}
\end{center}
\end{figure}

In a second scenario we now consider a more complex setting which is inspired by the real data application presented in Section \ref{sec5}. Precisely we assume that
\begin{eqnarray}
    \label{sim:scen2}
    Y_1&=&1+X_{12}+X_{13}+X_{13}^2+X_{14}+U_1,\nonumber\\ Y_2&=&1.7+X_{22}+X_{23}+X_{23}^2+X_{24}+U_2,
\end{eqnarray}
where the regressors $(X_{\ell 2},X_{\ell 2})^\top $ are bivariate normal distributed, that is  
$$(X_{\ell 2},X_{\ell 3})\sim \mathcal{N}_2(\textbf{0},\Sigma)\text{,~~~ with } \Sigma=  \sigma^2 \cdot \begin{pmatrix} 1 & 0.5 \\ 0.5 & 1 \end{pmatrix}, 
$$
$\sigma\in\{0.5,1,2\}$,
$X_{\ell 4}$ is binomial distributed, that is
$X_{\ell 4}\sim \Bin(n_\ell,0.5),$ 
and the errors $U_\ell$ are independent  centered normal distributed with variance 
$\sigma^2$. As a consequence, the link function $\Lambda$ is the distribution function of a normal distribution with mean $0$ and variance $\sigma^2$.
For a given $x=(1,1,1,1,1)$ we obtain the following underlying true values
\begin{equation}\label{sim:scen2_values}
|| F^1_{Y|X}(\cdot |(1,1,1,1,1)) - F^2_{Y|X}( \cdot |(1,1,1,1,1)) ||^2_2= 
\begin{cases}
    0.256 \text{  for $\sigma^2=0.25$}\\
    0.135 \text{  for $\sigma^2=1$}\\
    0.066 \text{  for $\sigma^2=4$}
\end{cases}
\end{equation}
for $y \in I=[2,8]$ and we denote this scenario by Scenario 2a. A visualization of the true underlying conditional distribution functions for $\sigma=1$ is given in   the upper left panel of Figure \ref{fig:scen2}.
Compared to model \eqref{sim:scen1}, the model \eqref{sim:scen2} is more complex and the estimation of the corresponding  parameters requires  larger sample sizes   to obtain reasonable results. Motivated by the size of the real data set in Section \ref{sec4}, we now consider sample sizes given by $n=1000,2000,5000,10000$ for the simulations.

In the upper part  of Figure \ref{fig:scen2} we display the empirical rejection probabilities for the medium variance of $\sigma^2=1$ and the large variance of $\sigma^2=4$, respectively. Again, the vertical red line indicates the margin of the null hypothesis in \eqref{hypotheses}. We observe that, as expected, the power increases with increasing sample sizes, while the approximation of the level is very precise for all sample sizes under consideration. We conclude that the test is able to detect relevant differences with a high degree of accuracy, even when the underlying models are complex and involve the estimation of many parameters.

\begin{figure}[h]
\begin{center}
\includegraphics[width=0.32\textwidth]{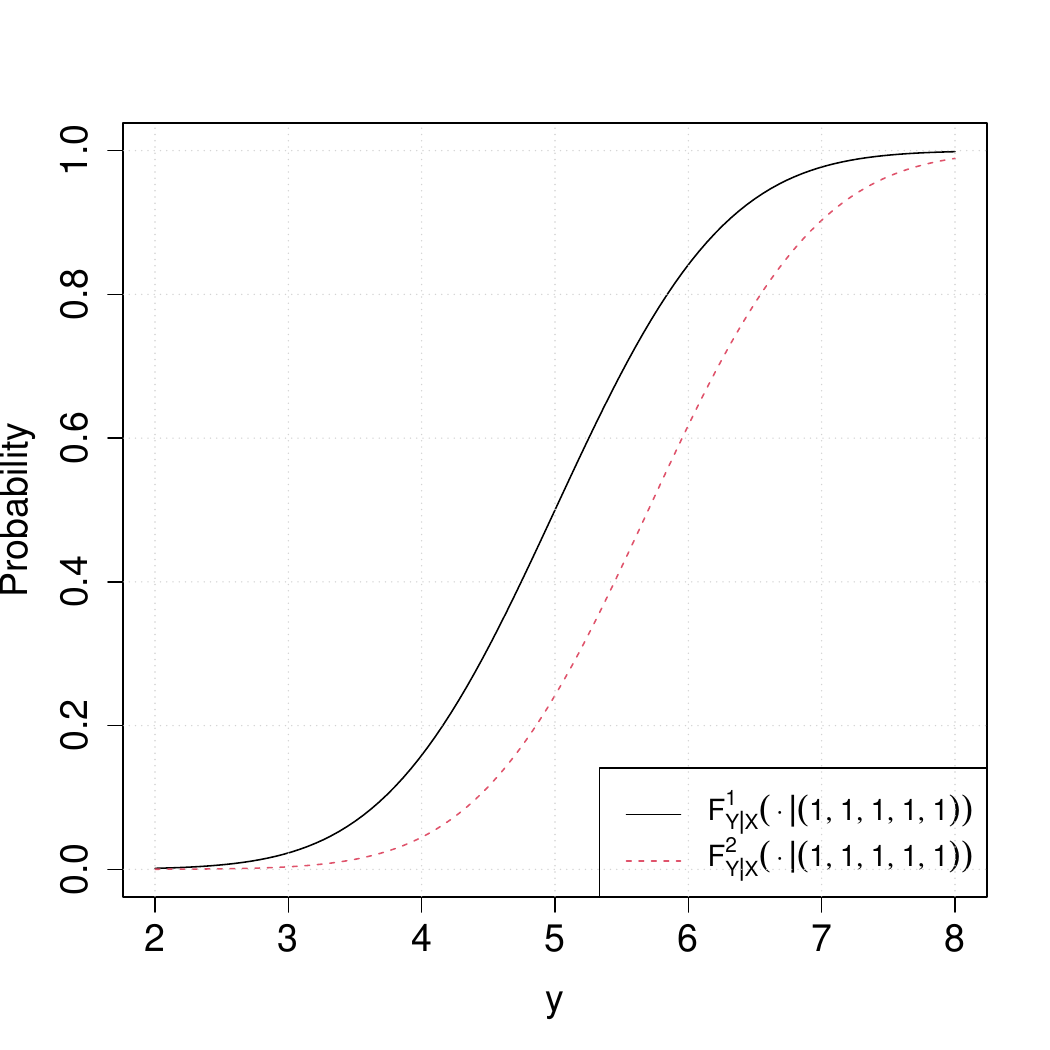}
\includegraphics[width=0.32\textwidth]{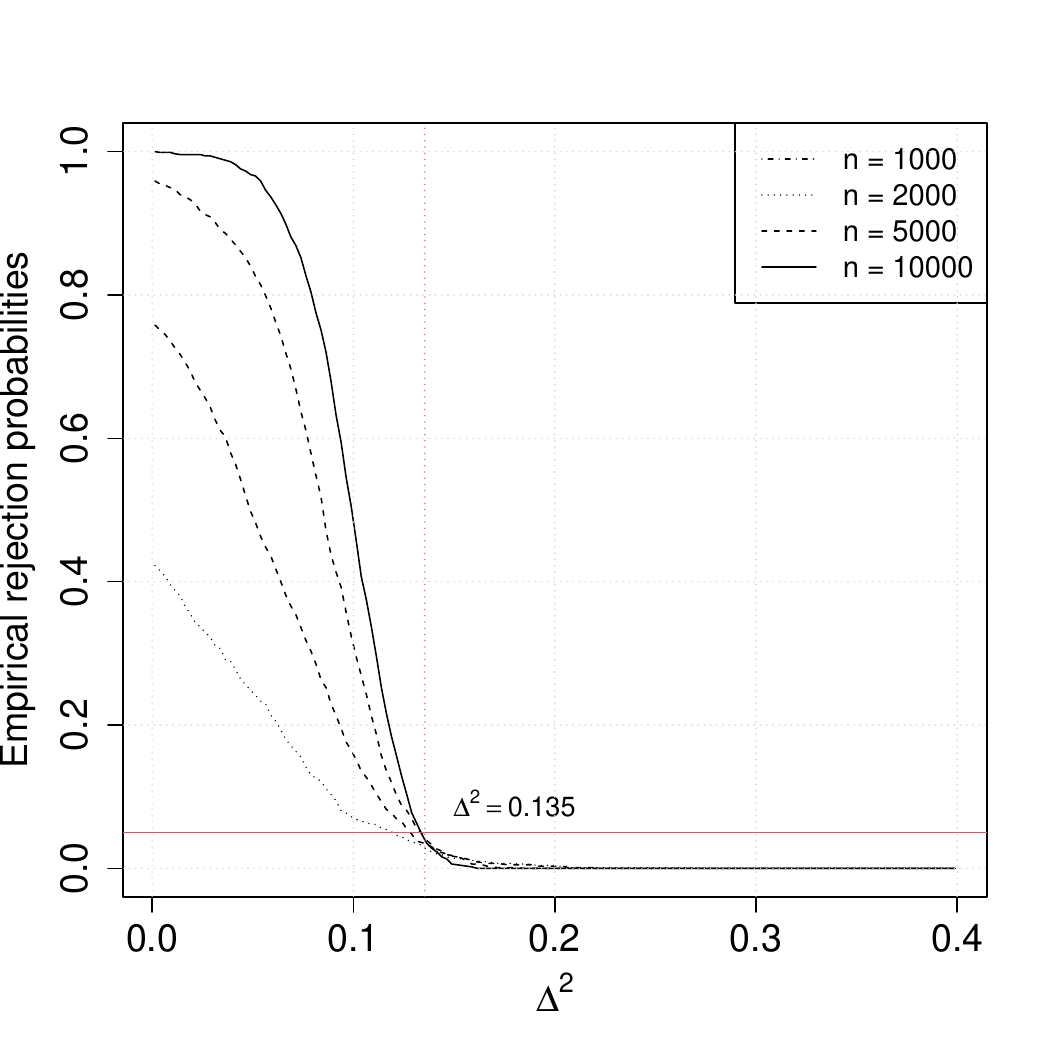}
\includegraphics[width=0.32\textwidth]{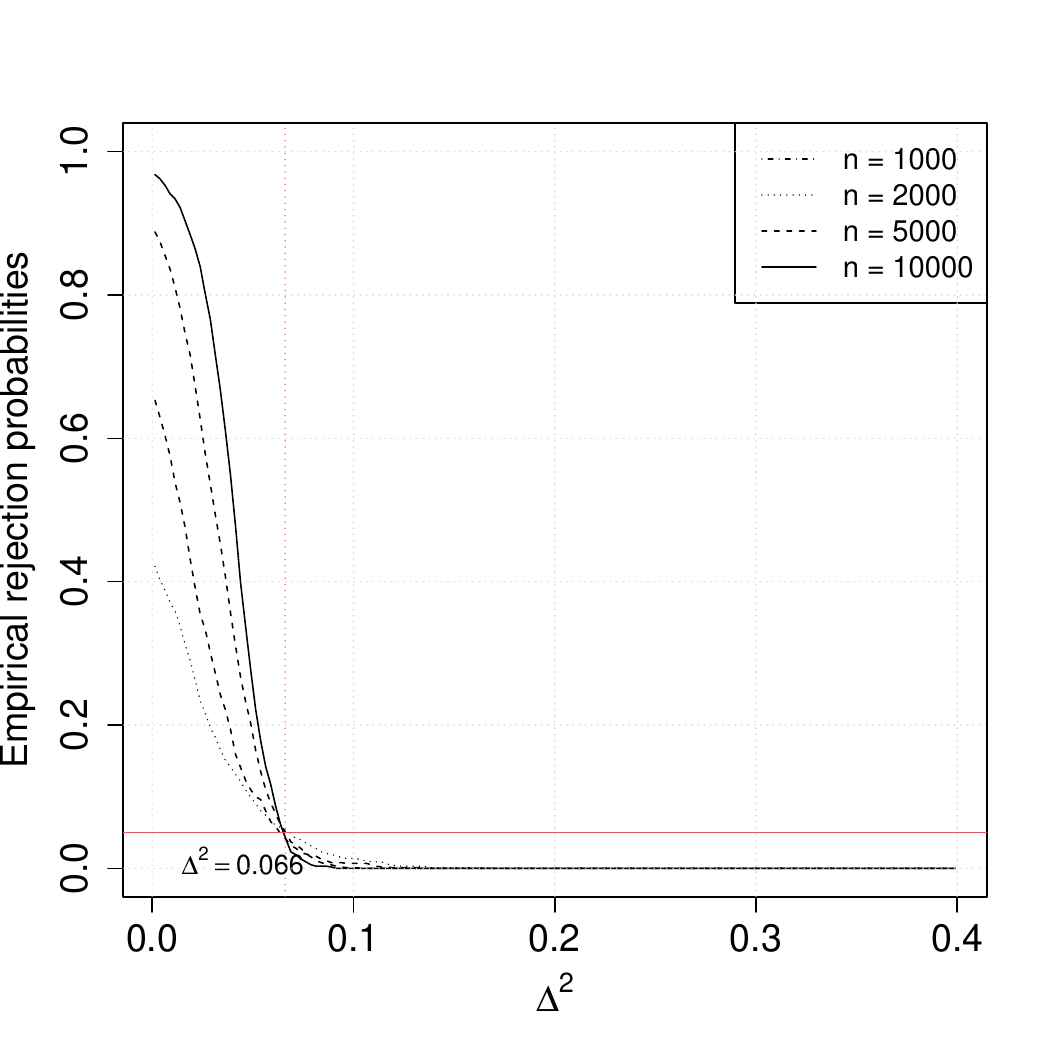}
\includegraphics[width=0.32\textwidth]{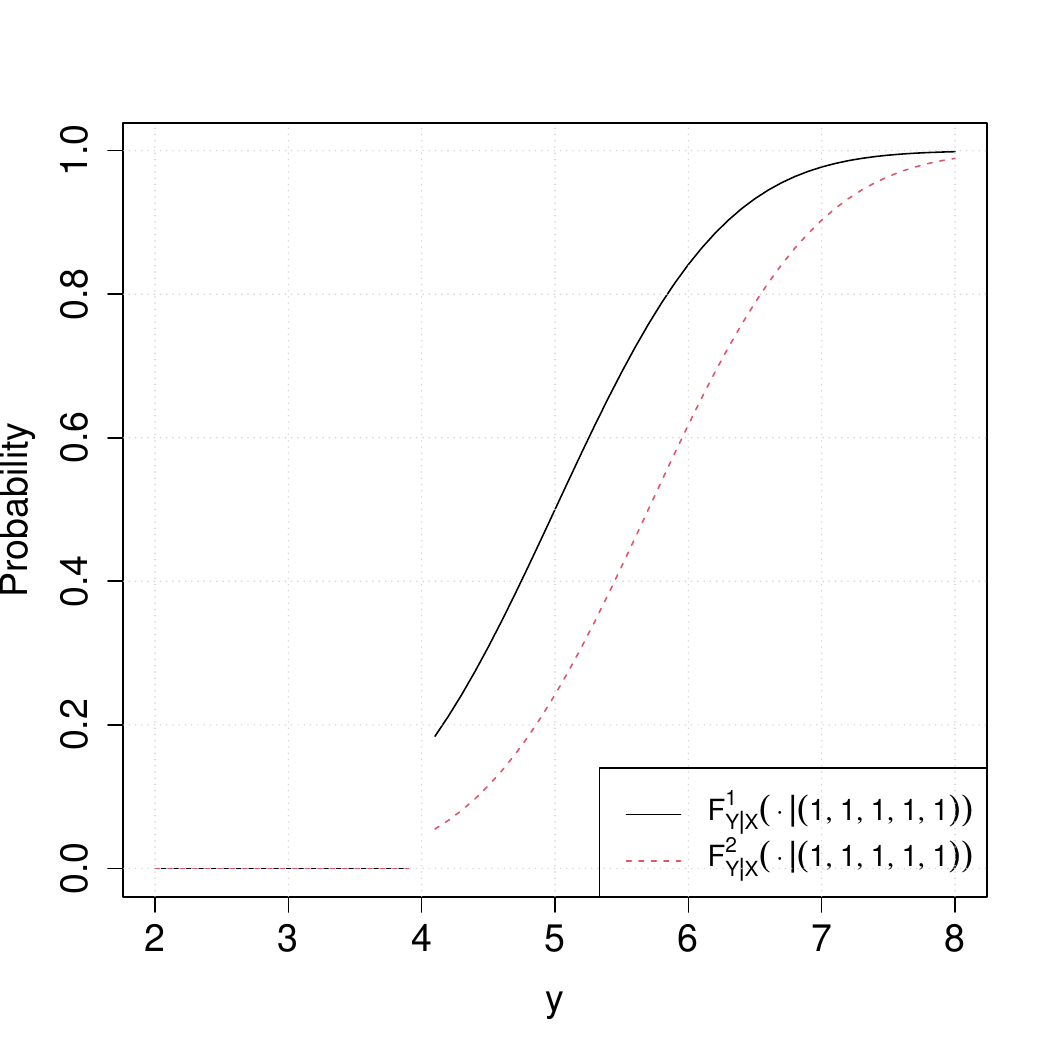}
\includegraphics[width=0.32\textwidth]{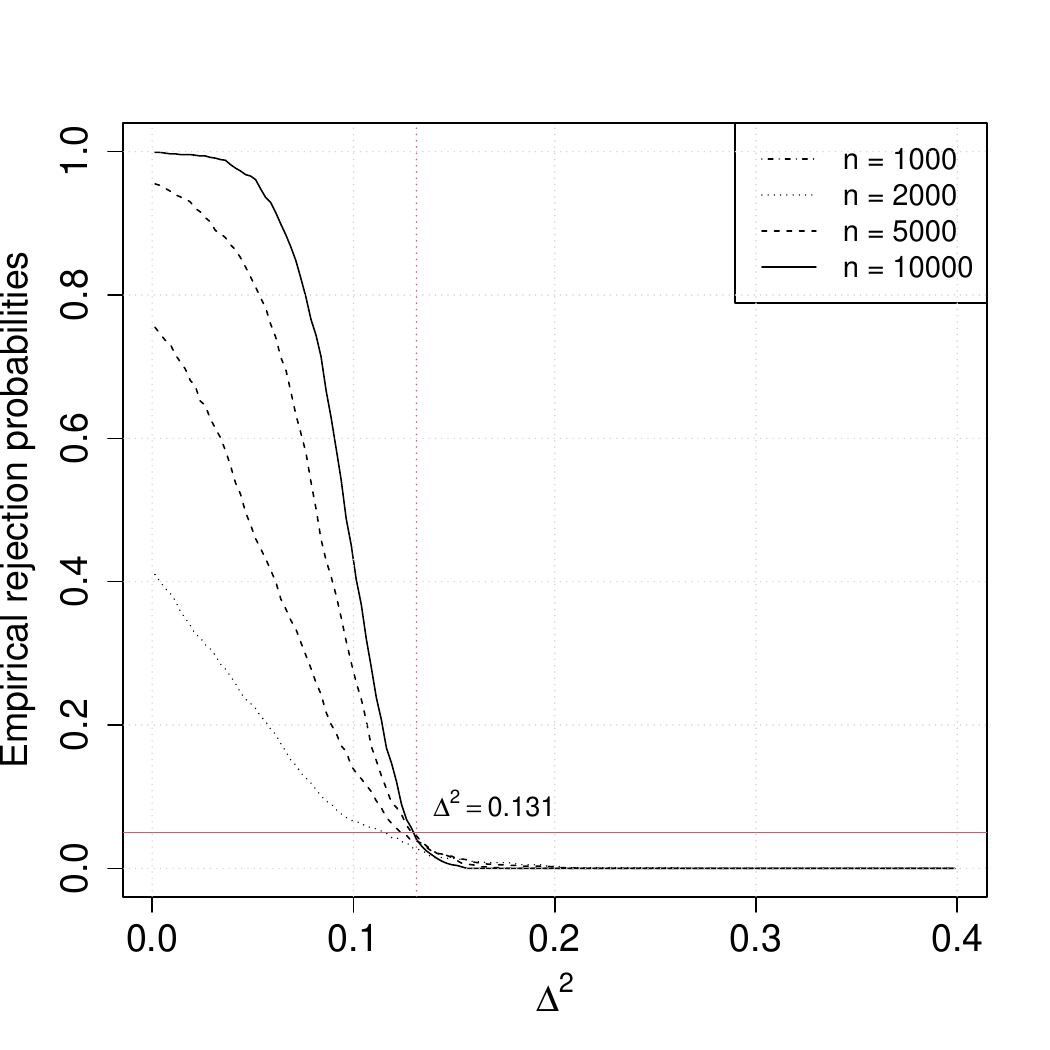}
\includegraphics[width=0.32\textwidth]{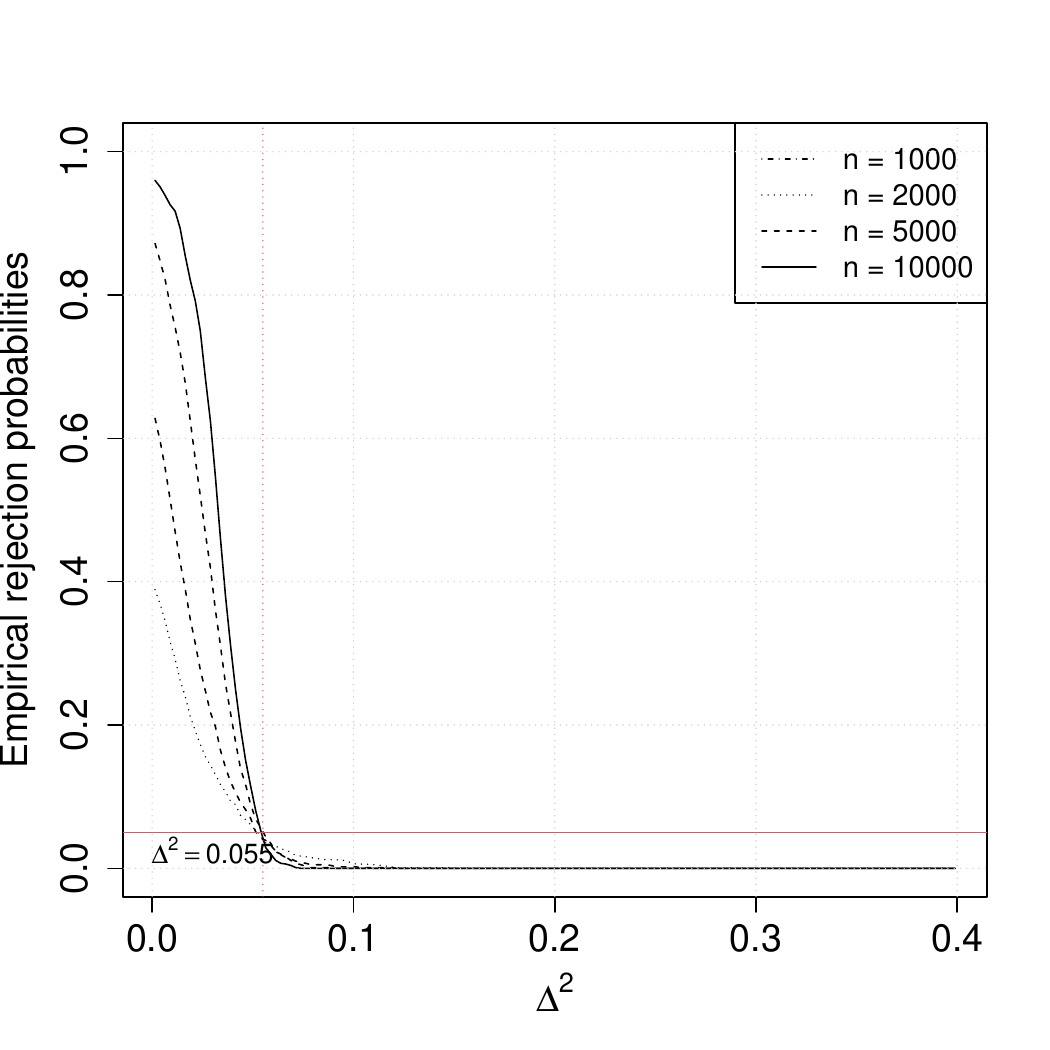}
\caption{
Left column: The true underlying conditional distribution functions 
$F^1_{Y|X}(\cdot |(1,1,1,1,1)) $ and 
$F^2_{Y|X}(\cdot |(1,1,1,1,1))$ for the model \eqref{sim:scen2} in Scenario 2a (upper row) and Scenario 2b (lower row) assuming $\sigma=1$. 
Center and right column: Empirical rejection probabilities of the test \eqref{dec_rule} for different sample sizes, $\epsilon=0.1$, in dependence of the threshold $\Delta^2$  for $\sigma=1$ (center column) and $\sigma=2$ (right column). The red solid line indicates the significance level of $\alpha=0.05$, the red dashed line indicates the true underlying value and thus the margin of the null hypothesis in \eqref{hypotheses}.
}
\label{fig:scen2}
\end{center}
\end{figure}

Finally, we consider the model
\begin{align}
    \label{det111}
    \tilde Y_\ell  = \max(Y_\ell,4)   ~~(\ell =1,2), 
\end{align}
where $Y_1$ and $Y_2$ are defined in \eqref{sim:scen2}, keeping the distributions of the regressors and the error terms as stated above. This scenario is denoted by Scenario 2b and describes the presence of a jump  with moderate size in the distribution function (see the lower left panel of Figure \ref{fig:scen2}). Such a phenomenon might occur in income data (see \citealp{wied:2024}),and  violates Assumption 1.3. Our simulations show that this violation is not problematic in the types of applications we are interested in.
For a given $x=(1,1,1,1,1)$ we now obtain the following underlying true values
\begin{equation}\label{sim:scen3_values}
|| F^1_{Y|X}(\cdot |(1,1,1,1,1)) - F^2_{Y|X}(\cdot|(1,1,1,1,1)) ||^2_2= 
\begin{cases}
    0.256 \text{  for $\sigma^2=0.25$}\\
    0.131 \text{  for $\sigma^2=1$}\\
    0.055 \text{  for $\sigma^2=4$}
\end{cases}
\end{equation}
for $y \in I=[2,8]$.
The lower row of Figure \ref{fig:scen2}  shows the empirical rejection probabilities for the medium variance of $\sigma^2=1$ (middle column) and the high variance of $\sigma^2=4$ (right column). Compared to Scenario 2a the results are almost exactly the same and the differences are barely visible.
The simulated type I error rates of the test \eqref{dec_rule} at the margin of the hypothesis \eqref{dec_rule}
for 
Scenarios 2a and 2b are  shown in second and third row of Table \ref{tab:level1}, where we consider  different choices 
$\epsilon \in \{0.05,0.1,0.2\}$
for the constant $\epsilon$ in the self-normalizing statistic \eqref{det96}. 
As in  Scenario 1, these results demonstrate 
 that  the choice of $\epsilon$ only slightly influences the performance of the test. 
Furthermore, investigating the effect of the jump point in the distribution function, i.e. comparing Scenario 2a and Scenario 2b, we observe that differences are only visible for the configurations with $\sigma^2=1$ and $\sigma^2=4$. However, even in these cases the results are very similar and the differences are negligible. 

\begin{table}[]
\centering
\begin{tabular}{ll|lll|lll|lll}
    \multicolumn{2}{c}{} & \multicolumn{3}{c}{$\sigma^2=0.25$} & \multicolumn{3}{c}{$\sigma^2=1$} & \multicolumn{3}{c}{$\sigma^2=4$}  \\
    \multicolumn{2}{c}{} & \multicolumn{3}{c}{$\epsilon$} & \multicolumn{3}{c}{$\epsilon$} & \multicolumn{3}{c}{$\epsilon$}  \\
    & $n$ & $0.05$ & $0.1$ & $0.2$ & $0.05$ & $0.1$ & $0.2$ & $0.05$ & $0.1$ & $0.2$ \\ \hline
\multirow{3}{*}{Scen. 1}     & 200       & 0.032      & 0.027    & 0.033     & 0.045     & 0.040     & 0.041     & 0.039           & 0.051          & 0.035          \\
 & 500       & 0.064           & 0.059           & 0.064   & 0.067  & 0.062     & 0.059      & 0.077      & 0.068      & 0.050    \\
  & 1000      & 0.043       & 0.038        & 0.039        & 0.046       & 0.034     & 0.037   & 0.050    & 0.052     & 0.045       \\\hline
\multicolumn{1}{c}{\multirow{4}{*}{Scen. 2a}} & 1000   & 0.030    & 0.026  & 0.027   & 0.030   & 0.030   & 0.036  & 0.048  & 0.054 & 0.041     \\
\multicolumn{1}{c}{}          & 2000      & 0.033    & 0.032    & 0.047    & 0.043  & 0.035  & 0.047    & 0.054  & 0.042    & 0.044          \\
\multicolumn{1}{c}{}           & 5000      & 0.037   & 0.050     & 0.052   & 0.043  & 0.044  & 0.040    & 0.043   & 0.052    & 0.044          \\
\multicolumn{1}{c}{}          & 10000     & 0.044   & 0.043  & 0.047  & 0.049   & 0.040   & 0.056    & 0.051    & 0.042   & 0.043      \\\hline
\multirow{4}{*}{Scen. 2b}    
& 1000      & 0.029   & 0.026    & 0.027     & 0.030     & 0.027     & 0.038    & 0.049   & 0.040   & 0.038     \\
& 2000    & 0.033    & 0.032        & 0.047    & 0.039    & 0.038    & 0.052     & 0.069   & 0.042    & 0.039          \\
& 5000      & 0.037    & 0.050     & 0.052   & 0.046    & 0.045    & 0.040    & 0.053   &   0.051   & 0.047          \\
& 10000     & 0.044       & 0.043    & 0.047    & 0.500  & 0.040    & 0.056    & 0.057     &  0.043    & 0.042         
\end{tabular}
\caption{Simulated type I error rates of the test \eqref{dec_rule} at the margin 
$|| F^1_{Y|X}(\boldsymbol{\cdot}|x) - F^2_{Y|X}(\boldsymbol{\cdot}|x) ||_2 =  \Delta $  of the hypothesis \eqref{hypotheses} for different variance levels $\sigma^2$, different sample sizes and different choices of the parameter $\epsilon$ 
in the self-normalizing statistic \eqref{det96}-
Scenario 1, 2a and 2b  are defined in \eqref{sim:scen1}, \eqref{sim:scen2} and \eqref{det111}, respectively.}
\label{tab:level1}
\end{table}

\section{Application to German income data}\label{sec5}
We apply our new test for detecting relevant changes in the conditional income distribution of German employees based on micro data from the German SOEP (\citealp{soep:2022}) between the years 2013 and 2020. The dependent variable is given by $log(income)^\ell$, $\ell \in \{2013,2020\}$, while the linear predictor is given by Mincer's earning equation
\begin{equation*}
pred^\ell(y) := \beta_0(y) + \beta^\ell_1(y)\ educ^\ell + \beta^\ell_2(y)\ exper^\ell + \beta^\ell_3(y)\ {exper^\ell}^2 + \beta^\ell_4(y)\ partt^\ell.
\end{equation*}
This means that we model the expression
\begin{equation}\label{baselinemodel}
P(log(income)^\ell \leq y | educ^\ell, exper^\ell, partt^\ell) = \Phi\left(pred^\ell(y) \right),
\end{equation}
where $\Phi$ denotes the distribution function of the standard normal distribution. We compare the conditional distribution functions of the logarithmic monthly income given the covariates {\it years of education}, {\it years of working experience}, {\it squared years of working experiences}, {\it part time job yes/no} in the years 2013 and 2020. The sample sizes are given by $16031$ for the year 2020 and $18262$ for the year 2013. We choose this particular time interval to simplify the interpretation of the results: there was a rise in the salary for a {\it Minijob} in 2013 and the Covid pandemic led to more substantial changes in the income structure after 2020.

The goal is to answer the question in which sense the incomes have increased statistically significantly more than the consumer prices. Economists are typically interested in how the link between incomes (or wages) and inflation looks like. For example, \citet{gordon:1988} showed empirically that the link is less strong than expected, while \citet{jorda:2023} show that the link might have become stronger after the pandemic in 2020. In fact, for our data, there is some empirical evidence that incomes increased to a higher extent than consumer prices, with the implication that there are other important drivers of inflation. In general, incomes increased by $17\%$ (\citet{sozialpolitik:2025}), while inflation was around $7.4\%$ (\citet{destatis:2025}), as Table \ref{table:percentages} shows. However, these numbers do not yield information about changes in the whole (conditional) distribution of income.

\begin{table}\label{table:percentages}
\centering
\begin{tabular}{c|ccccccc}
Year & 2014 & 2015 & 2016 & 2017 & 2018 & 2019 & 2020 \\\hline
Income increase &2.9 &3.0 &2.5 &2.7 &3.2 &3.0 &-0.1 \\\hline
Inflation &1.0 &0.5 &0.5 &1.5 &1.8 &1.4 &0.5
\end{tabular}
\caption{Income increase and inflation in percentages in Germany in a particular year compared to the year before}
\end{table}

For the application of the test \eqref{dec_rule} in the present context, we have to specify a value of the threshold  $\Delta$. Our approach considers the $L^2$-distance between  the two conditional distribution functions  (on the interval $I= [2,10]$) corresponding to the years  $2013$ and $2020$ (for a specified set of regressor values) under the assumption that the change of the incomes is exactly the same as the change of the consumer prices. Under this assumption, all coefficients in \eqref{baselinemodel} remain the same except the intercept, which increases by $0.074$. We are interested in the results for three cases: 
$x^\top =(q^{edu}_{0.1}, q^{exper}_{0.1}, m )$,  $x^\top=(q^{edu}_{0.5}, q^{exper}_{0.5}, m^{partt} )$ and $x^\top=(q^{edu}_{0.9}, q^{exper}_{0.9}, m^{partt} )$, where
$q^{edu}_{\alpha}$ and  $q^{exper}_{\alpha}$ denote the empirical $\alpha$-quantile of the regressors 
{\it educ} and {\it exper}, respectively,   
and $ m^{partt} )$  denotes  the empirical mean of the component {\it partt}.
For the three different cases, the values for $\Delta^2$ are chosen as $0.001$, $0.0008$ and $0.0009$, respectively.

Let us first discuss the calculation of the test results. We consider $\epsilon=0.1$, and for this choice  the quantile $q_{0.95}$ is given by $1.7546$. The values of the self-normalizing statistic $\hat V_n$ in \eqref{det96} are given by $0.0223$, $0.0236$, $0.0095$ (for the regressor quantiles of $10\%$, $50\%$ and $90\%$). For the different values of $\Delta^2$, we then obtain 
$$
\hat T_n=0.048 > 0.0419 = 0.001 + 1.7546 \cdot 0.0223$$ for the $10\%$-quantile ($p$-value 0.031), 
$$
\hat T_n=0.028 < 0.0422 = 0.0008 + 1.7546 \cdot 0.0236
$$ for the $50\%$-quantile ($p$-value 0.12) and 
$$
\hat T_n=0.028 > 0.0176 = 0.0009 + 1.7546 \cdot 0.0095
$$ for the $90\%$-quantile ($p$-value 0.012).

This means that, for the lower and upper quantile of the regressors, but not for the median, there is a statistically significant relevant change. Figure \ref{fig:application10}  illustrates the results. The upper panel shows the estimated conditional distribution functions for the three quantiles and the two years. The lower panel shows the $p$-values of the test \eqref{dec_rule} for the different  cases as a function of $\Delta^2$. This curve increases in $\Delta^2$, which is obvious given the construction of the test. For the lower and upper quantile of the regressors, the $p$-values are below $0.05$ for some values of $\Delta^2$. For the median, this is never the case. The estimates of $\hat \Delta_\alpha$ in \eqref{det113}, which yield measures of evidence against the null hypothesis, are given by $0.0097$ for the lower quantile, $0.0129$ for the upper quantile and $0$ for the median. This suggests that the ``more extreme'' employees faced a stronger increase in income compared to the ``moderate'' ones.

\begin{figure}[h]
\begin{center}
\includegraphics[width=0.32\textwidth]{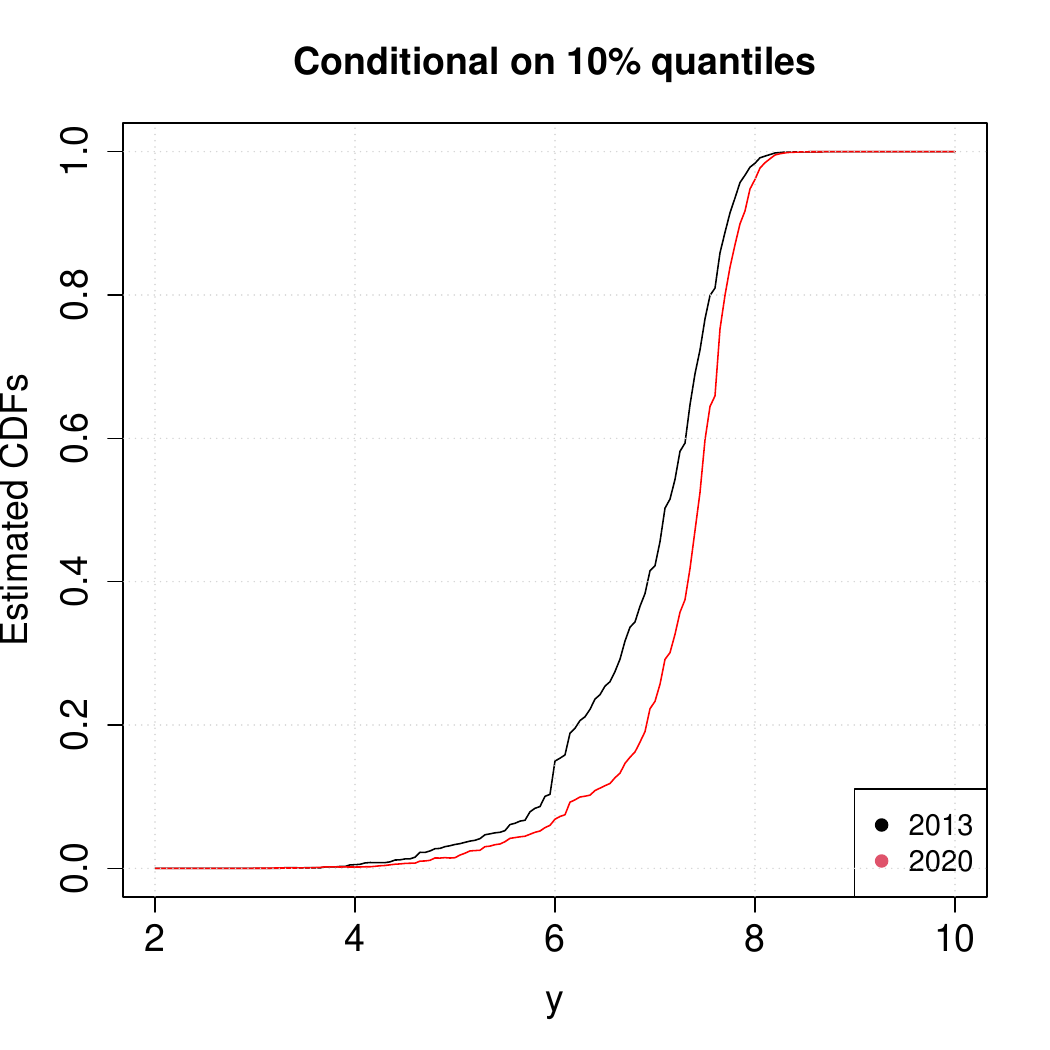}
\includegraphics[width=0.32\textwidth]{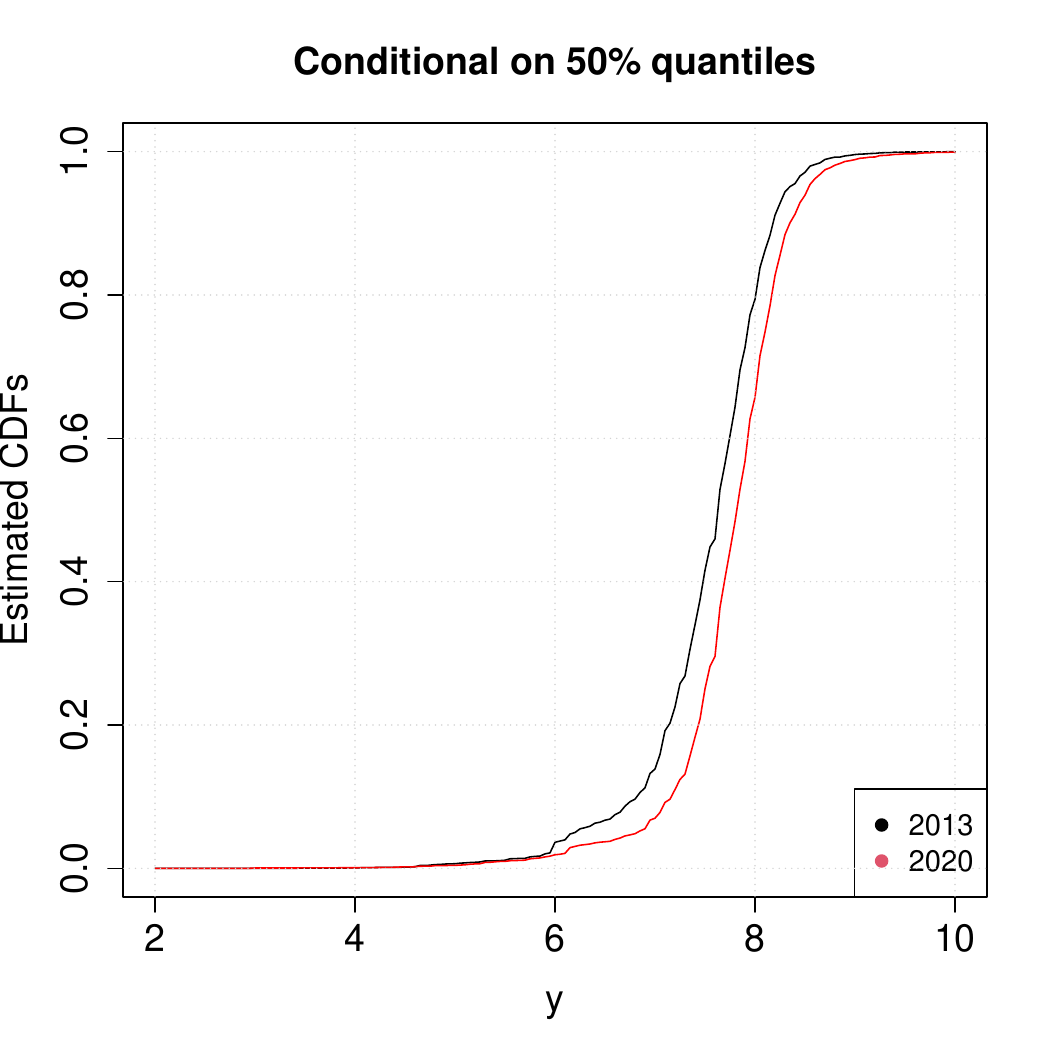}
\includegraphics[width=0.32\textwidth]{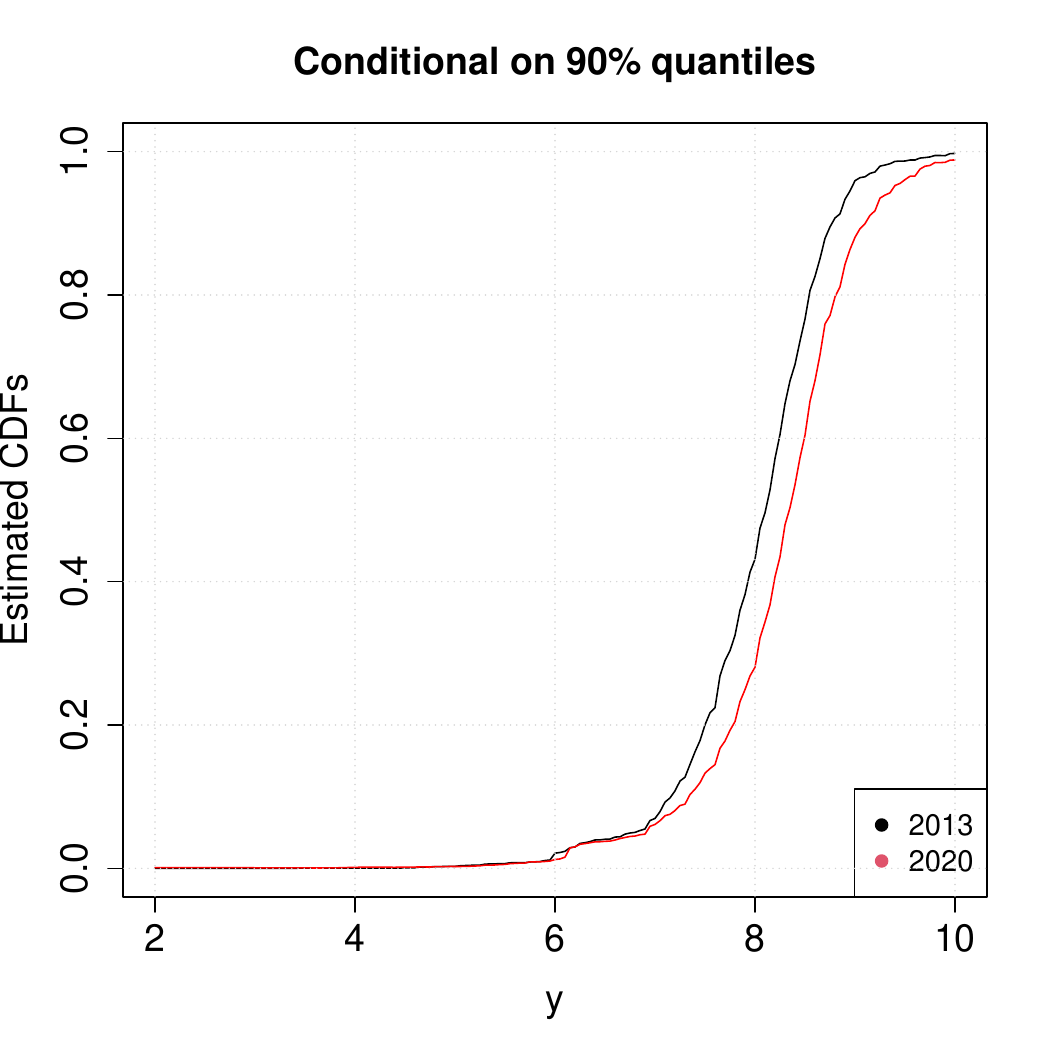}
\includegraphics[width=0.32\textwidth]{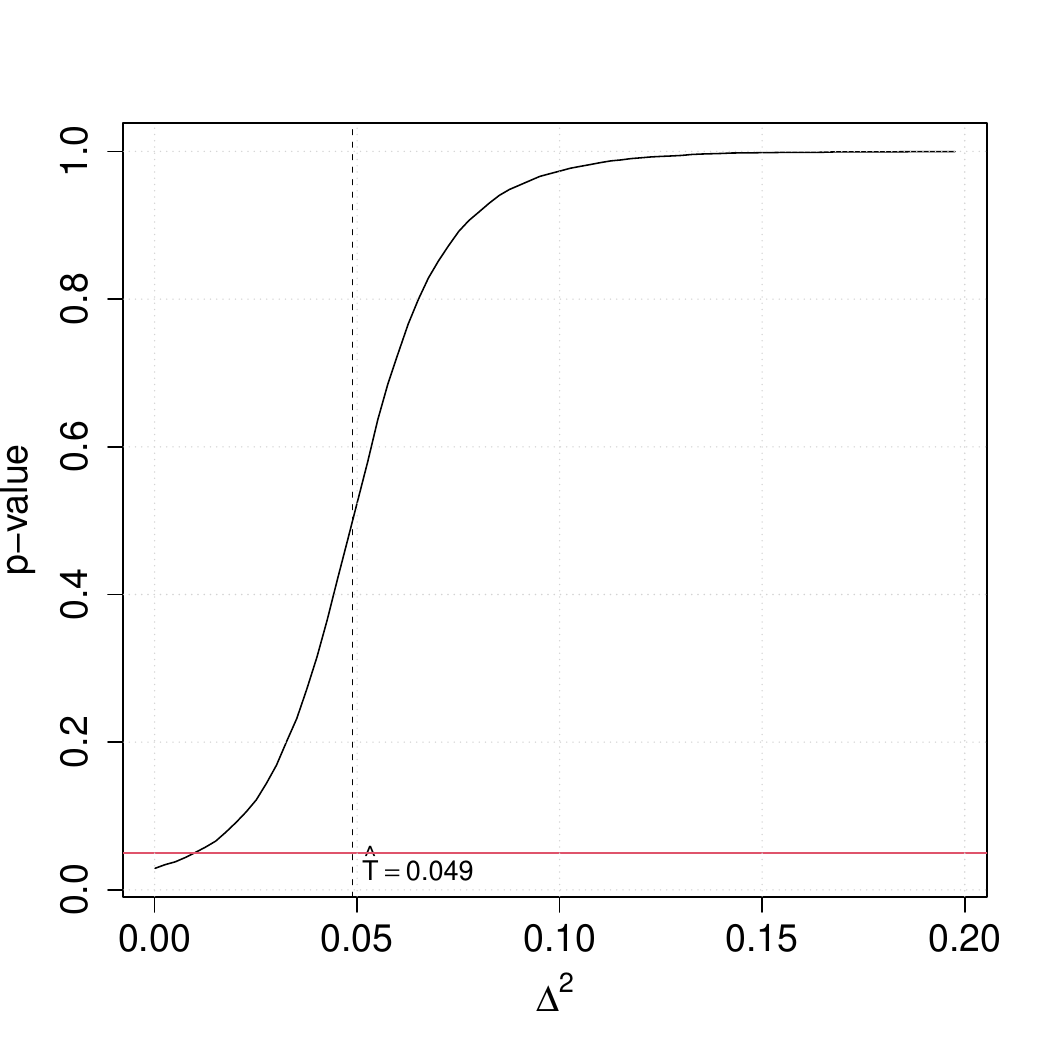}
\includegraphics[width=0.32\textwidth]{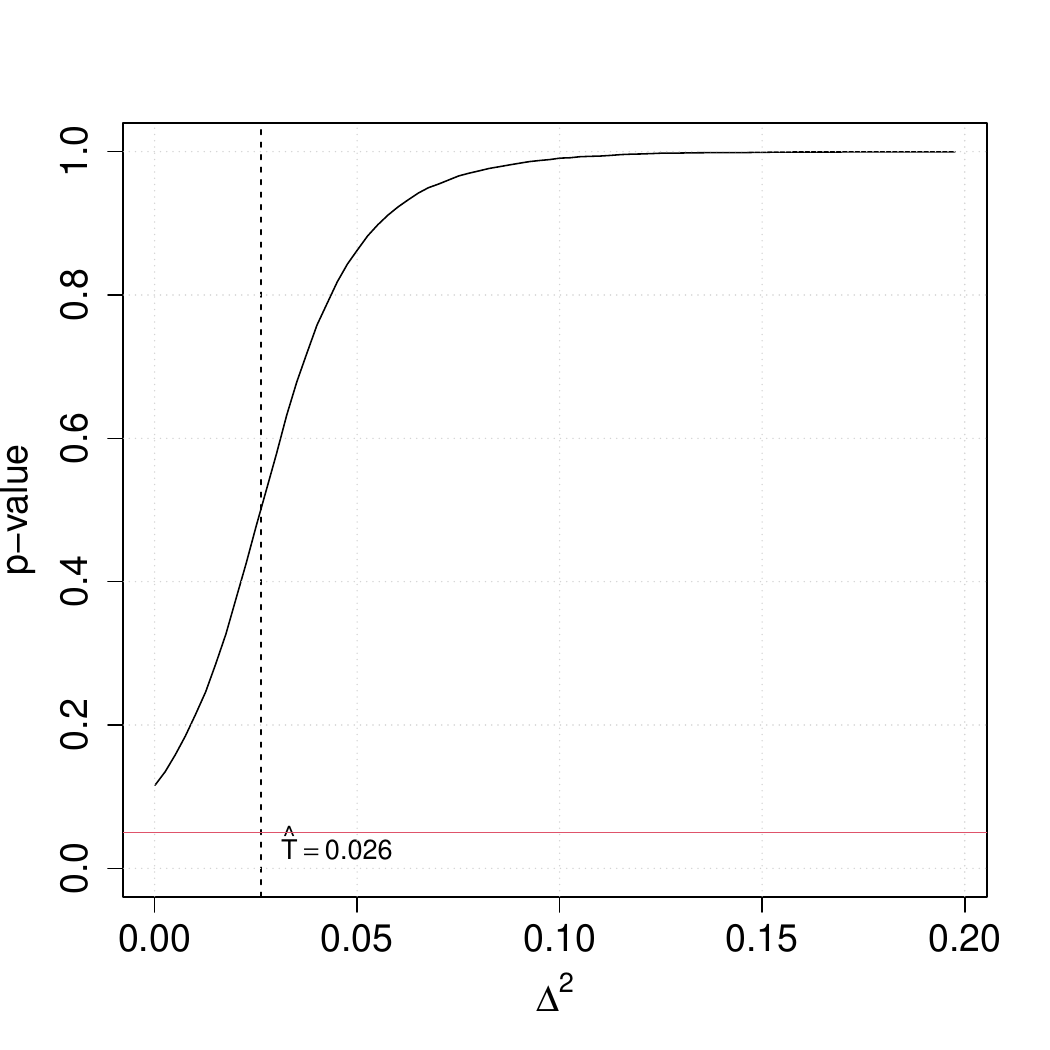}
\includegraphics[width=0.32\textwidth]{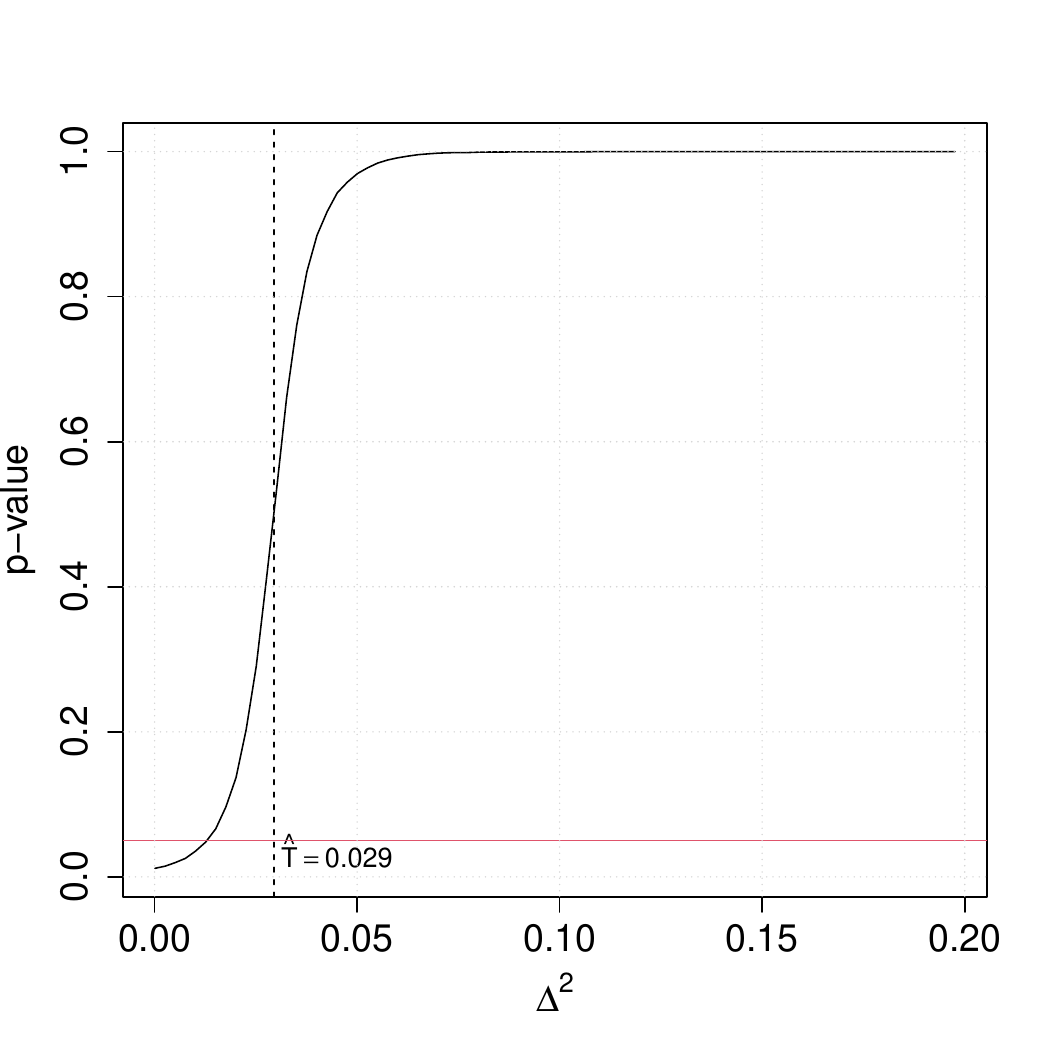}
\caption{Upper row: The estimated cumulative distribution functions (conditional distribution functions) for the 2013 and 2020 data. Lower row: $p$-values of the test \eqref{dec_rule} in dependence of the threshold $\Delta^2$. The red solid line indicates the significance level of $\alpha=0.05$, the black dashed line indicates the value of the test statistic.}
\label{fig:application10}
\end{center}
\end{figure}


\bigskip

{\bf Acknowledgements} This work was supported by TRR 391 \textit{Spatio-temporal Statistics for the Transition of Energy and Transport} (Project number 520388526) funded by the Deutsche Forschungsgemeinschaft (DFG, German Research Foundation).

\bibliographystyle{ecta}
\bibliography{sample}

\begin{thebibliography}{30}
\newcommand{\enquote}[1]{``#1''}
\expandafter\ifx\csname natexlab\endcsname\relax\def\natexlab#1{#1}\fi

\bibitem[\protect\citeauthoryear{Andrews}{Andrews}{1997}]{andrews:1997}
\textsc{Andrews, D.} (1997): \enquote{A Conditional Kolmogorov Test,} \emph{Econometrica}, 65(5), 1097--1128.

\bibitem[\protect\citeauthoryear{Berger and Delampady}{Berger and Delampady}{1987}]{berger1987}
\textsc{Berger, J.~O. and M.~Delampady} (1987): \enquote{{Testing Precise Hypotheses},} \emph{Statistical Science}, 2, 317 -- 335.

\bibitem[\protect\citeauthoryear{Bradley}{Bradley}{2007}]{Bradley.2007}
\textsc{Bradley, R.~C.} (2007): \emph{Introduction to Strong Mixing Conditions. Vol 1-3}, Heber City, Utah: Kendrick.

\bibitem[\protect\citeauthoryear{Brise{\~n}o-Sanchez, Hohberg, Groll, and Kneib}{Brise{\~n}o-Sanchez et~al.}{2020}]{sanchez:2020}
\textsc{Brise{\~n}o-Sanchez, G., M.~Hohberg, A.~Groll, and T.~Kneib} (2020): \enquote{Flexible Instrumental Variable Distributional Regression,} \emph{Journal of the Royal Statistical Society, Series A}, 183(4), 1553--1574.

\bibitem[\protect\citeauthoryear{Chernozhukov, Fernández-Val, and Luo}{Chernozhukov et~al.}{2025+}]{chernozhukov:2025}
\textsc{Chernozhukov, V., I.~Fernández-Val, and S.~Luo} (2025+): \enquote{Distribution Regression with Sample Selection and UK Wage Decomposition,} \emph{Journal of Political Economy}, forthcoming.

\bibitem[\protect\citeauthoryear{Chernozhukov, Fernández-Val, and Melly}{Chernozhukov et~al.}{2013}]{chernozhukov:2013}
\textsc{Chernozhukov, V., I.~Fernández-Val, and B.~Melly} (2013): \enquote{Inference on Counterfactual Distributions,} \emph{Econometrica}, 81(6), 2205--2268.

\bibitem[\protect\citeauthoryear{Chernozhukov, Fernández-Val, Newey, Stouli, and Vella}{Chernozhukov et~al.}{2020}]{chernozhukov:2020}
\textsc{Chernozhukov, V., I.~Fernández-Val, W.~Newey, S.~Stouli, and F.~Vella} (2020): \enquote{Semiparametric Estimation of Structural Functions in Nonseparable Triangular Models,} \emph{Quantitative Economics}, 11(2), 503--533.

\bibitem[\protect\citeauthoryear{Delgado, García-Suaza, and Sant'Anna}{Delgado et~al.}{2022}]{delgado:2022}
\textsc{Delgado, M., A.~García-Suaza, and P.~Sant'Anna} (2022): \enquote{Distribution Regression in Duration Analysis: An Application to Unemployment Spells,} \emph{Econometrics Journal}, 25(3), 675--698.

\bibitem[\protect\citeauthoryear{{Destatis}}{{Destatis}}{2025}]{destatis:2025}
\textsc{{Destatis}} (2025): \enquote{Verbraucherpreisindex und Inflationsrate,} \emph{https://www.destatis.de/DE/Themen/Wirtschaft/Preise/Verbraucherpreisindex/ \_inhalt.html (07.04.2025)}.

\bibitem[\protect\citeauthoryear{Dette and Schumann}{Dette and Schumann}{2024}]{detteschumann:2024}
\textsc{Dette, H. and M.~Schumann} (2024): \enquote{Testing for Equivalence of Pre-Trends in Difference-in-Differences Estimation,} \emph{Journal of Business and Economic Statistics}, 42(4), 1289--1301.

\bibitem[\protect\citeauthoryear{Foresi and Peracchi}{Foresi and Peracchi}{1995}]{foresiperacchi:1995}
\textsc{Foresi, S. and F.~Peracchi} (1995): \enquote{The Conditional Distribution of Excess Returns: An Empirical Analysis,} \emph{Journal of the American Statistical Association}, 90, 451--466.

\bibitem[\protect\citeauthoryear{Gordon}{Gordon}{1988}]{gordon:1988}
\textsc{Gordon, R.} (1988): \enquote{The Role of Wages in the Inflation Process,} \emph{American Econonomic Review}, 78(2), 276--283.

\bibitem[\protect\citeauthoryear{H{\"o}rmann and Kokoszka}{H{\"o}rmann and Kokoszka}{2010}]{HrmannKok}
\textsc{H{\"o}rmann, S. and P.~Kokoszka} (2010): \enquote{{Weakly dependent functional data},} \emph{The Annals of Statistics}, 38, 1845 -- 1884.

\bibitem[\protect\citeauthoryear{Hu and Lei}{Hu and Lei}{2024}]{hulei:2024}
\textsc{Hu, X. and J.~Lei} (2024): \enquote{A Two-Sample Conditional Distribution Test Using Conformal Prediction and Weighted Rank Sum,} \emph{Journal of the American Statistical Association}, 119(546), 1136--1154.

\bibitem[\protect\citeauthoryear{Jordà and Nechio}{Jordà and Nechio}{2023}]{jorda:2023}
\textsc{Jordà, O. and F.~Nechio} (2023): \enquote{Inflation and Wage Growth since the Pandemic,} \emph{European Econonomic Review}, 156, 104474.

\bibitem[\protect\citeauthoryear{Kneib, Silbersdorff, and Säfken}{Kneib et~al.}{2023}]{kneib:2023}
\textsc{Kneib, T., A.~Silbersdorff, and B.~Säfken} (2023): \enquote{Rage Against the Mean - A Review of Distributional Regression Approaches,} \emph{Econometrics and Statistics}, 26, 99--123.

\bibitem[\protect\citeauthoryear{Kutta and Dette}{Kutta and Dette}{2024}]{kuttadette}
\textsc{Kutta, T. and H.~Dette} (2024): \enquote{Validating approximate slope homogeneity in large panels,} \emph{Journal of Econometrics}, 246, 105898.

\bibitem[\protect\citeauthoryear{Rothe and Wied}{Rothe and Wied}{2013}]{RotheWied2013}
\textsc{Rothe, C. and D.~Wied} (2013): \enquote{Misspecification Testing in a Class of Conditional Distributional Models,} \emph{Journal of the American Statistical Association}, 108, 314--324.

\bibitem[\protect\citeauthoryear{Sheehy and Wellner}{Sheehy and Wellner}{1992}]{sheehy:1992}
\textsc{Sheehy, A. and J.~Wellner} (1992): \enquote{Uniform Donsker Classes of Functions,} \emph{Annals of Probability}, 20(4), 1983--2020.

\bibitem[\protect\citeauthoryear{{Sozialpolitik-aktuell}}{{Sozialpolitik-aktuell}}{2025}]{sozialpolitik:2025}
\textsc{{Sozialpolitik-aktuell}} (2025): \enquote{Sozialpolitik-aktuell,} \emph{https://www.sozialpolitik-aktuell.de/files/sozialpolitik-aktuell/\_Politikfelder/Einkommen-Armut/Datensammlung/PDF-Dateien/tabIII1.pdf (07.04.2025)}.

\bibitem[\protect\citeauthoryear{{Sozio-oekonomisches Panel (SOEP)}}{{Sozio-oekonomisches Panel (SOEP)}}{2022}]{soep:2022}
\textsc{{Sozio-oekonomisches Panel (SOEP)}} (2022): \enquote{Version 37, Daten der Jahre 1984-2020 (SOEP-Core v37, EU-Edition),} \emph{DOI: 10.5684/soep.core.v37eu}.

\bibitem[\protect\citeauthoryear{Spady and Stouli}{Spady and Stouli}{2025}]{spady:2025}
\textsc{Spady, R. and S.~Stouli} (2025): \enquote{Gaussian Transforms Modeling and the Estimation of Distributional Regression Functions,} \emph{Econometrica}, forthcoming.

\bibitem[\protect\citeauthoryear{Troster and Wied}{Troster and Wied}{2021}]{troster:2021}
\textsc{Troster, V. and D.~Wied} (2021): \enquote{A Specification Test of Dynamic Conditional Distributions,} \emph{Econometric Reviews}, 40(2), 109--127.

\bibitem[\protect\citeauthoryear{Tukey}{Tukey}{1991}]{tukey1991}
\textsc{Tukey, J.~W.} (1991): \enquote{The Philosophy of Multiple Comparisons,} \emph{Statistical Science}, 6, 100--116.

\bibitem[\protect\citeauthoryear{van~der Vaart and Wellner}{van~der Vaart and Wellner}{2023}]{VaartWellner2023}
\textsc{van~der Vaart, A. and J.~Wellner} (2023): \emph{Weak Convergence and Empirical Processes: With Applications to Statistics}, Springer Series in Statistics, Cham, Switzerland: Springer, 2nd ed.

\bibitem[\protect\citeauthoryear{Wang, Oka, and Zhu}{Wang et~al.}{2023}]{wang:2023}
\textsc{Wang, Y., T.~Oka, and D.~Zhu} (2023): \enquote{Bivariate Distribution Regression with Application to Insurance Data,} \emph{Insurance: Mathematics and Economics}, 113, 215--232.

\bibitem[\protect\citeauthoryear{Wellek}{Wellek}{2010}]{wellek2010testing}
\textsc{Wellek, S.} (2010): \emph{Testing Statistical Hypotheses of Equivalence and Noninferiority}, CRC Press.

\bibitem[\protect\citeauthoryear{Wied}{Wied}{2024}]{wied:2024}
\textsc{Wied, D.} (2024): \enquote{Semiparametric Distribution Regression with Instruments and Monotonicity,} \emph{Labour Economics}, 90, 102565.

\bibitem[\protect\citeauthoryear{Wu}{Wu}{2005}]{Wu2005}
\textsc{Wu, W.~B.} (2005): \enquote{Nonlinear system theory: Another look at dependence,} \emph{Proceedings of the National Academy of Sciences}, 102, 14150--14154.

\bibitem[\protect\citeauthoryear{Zhou, Zheng, and Zhang}{Zhou et~al.}{2017}]{zhou:2017}
\textsc{Zhou, W.-X., C.~Zheng, and Z.~Zhang} (2017): \enquote{Two-Sample Smooth Tests for the Equality of Distributions,} \emph{Bernoulli}, 23(2), 951--989.

\end{thebibliography}


\appendix
\section{Appendix}
   \def\theequation{A.\arabic{equation}}	
   \setcounter{equation}{0}

In this section we give  more details on the proof of Theorem \ref{thm4}. More precisely we derive in Theorem \ref{thm3} below the weak convergence in \eqref{det100}, which is the essential ingredient in the proof of Theorem \ref{thm4}.
Our first result provides a uniform Bahadur expansion for the estimators of the semiparametric distribution regression model.

\begin{theorem}\label{thm1}
Under Assumption \ref{ass:1}, it holds for $\ell=1,2$
\begin{equation}\label{theoremA1}
\sqrt{n_\ell}\big(\hat \beta ^{\ell}(t,y)- \beta ^{\ell}(y)\big)= -\tfrac{1}{\sqrt{n_\ell}} {C_\ell^{-1}  \over t}S^\ell_{\lfloor n_\ell t \rfloor}(\beta^\ell(y))+D_n(t,y)
\end{equation}
 where
$$S_k^\ell(\beta(y))=\sum_{i=1}^k R_i^\ell(\beta^\ell(y)),$$
\begin{equation*}
    R_i^\ell(\beta^\ell(y)) = X^\ell_i I_{\{Y_i^\ell\leq y\}}\frac{\lambda\big(X_i^{\ell^\top }\beta^\ell(y)\big)}{\Lambda\big(X_i^{\ell^\top }\beta^\ell(y)\big)} - X^\ell_i (1-I_{\{Y_i^\ell> y\}})\frac{\lambda\big(X_i^{\ell^\top }\beta^\ell(y)\big)}{\Lambda\big(X_i^{\ell^\top }\beta^\ell(y)\big)}, 
\end{equation*}
\begin{equation*}
C_\ell=\mathbb{E}\Big[\frac{\partial}{\partial \beta^\ell(y)} R_i^\ell\big(\beta^\ell(y)\big) \Big] \in \mathbb{R}^{p\times p}.
\end{equation*}
and the remainder satisfies 
$\sup_{(t,y)\in[\epsilon,1]\times I} |D_n(t,y)| \rightarrow 0$ for $n_\ell \rightarrow \infty$.
\end{theorem}
\begin{proof}
Using a Taylor approximation as in the proof of Theorem 5.2 in \citet{chernozhukov:2013}, we obtain
\begin{equation*}
\sqrt{n_\ell}\big(\hat \beta ^{\ell}(t,y)- \beta ^{\ell}(y)\big)= -\tfrac{1}{\sqrt{n_\ell}} \left(\frac{1}{n_\ell}\sum_{i=1}^{[n_\ell t]} \frac{\partial}{\partial \beta^\ell(y)} R_i^\ell\big(\beta^\ell(y)\big)\right)^{-1}S^\ell_{\lfloor n_\ell t \rfloor}(\beta^\ell(y))+D_n(t,y)
\end{equation*}
with
\begin{equation*}
D_n(t,y) = - \frac{1}{2} \left(\frac{1}{n_\ell}\sum_{i=1}^{[n_\ell t]} \frac{\partial}{\partial \beta^\ell(y)} R_i^\ell\big(\beta^\ell(y)\big)\right)^{-1} E_n(t,y)\left(\sqrt{n_\ell}\big(\hat \beta ^{\ell}(t,y)- \beta ^{\ell}(y)\big),\hat \beta ^{\ell}(t,y)- \beta ^{\ell}(y)\big)\right),
\end{equation*}
where $E_n(t,y)$ is a bilinear form, which contains the appropriate entries of the matrix 
\begin{equation*}
\frac{1}{n_\ell}\sum_{i=1}^{[n_\ell t]} \frac{\partial}{\partial \beta^\ell(y) \partial \beta^\ell(y)^\top } R_i^\ell\big(\tilde \beta^\ell(y)\big).
\end{equation*}
Here, $\tilde \beta^\ell(t,y)$ is some intermediate point  between $\hat \beta^\ell(t,y)$ and $\beta^\ell(y)$. By similar arguments as given in the proof of Theorem 5.2 in \citet{chernozhukov:2013}, it follows that
$\sup_{(t,y) \in [\epsilon,1] \times I} |\hat \beta^\ell(t,y) - \beta^\ell(y)| \rightarrow_p 0$ for $n \rightarrow \infty$. Therefore, the statement of Theorem \ref{thm1} follows from the boundedness conditions in Assumption 1.
    \end{proof}
We continue noting  that 
\begin{align}
\label{det101}
\left\{\tfrac{1}{\sqrt{n_\ell}}S^\ell_{\lfloor n_\ell t \rfloor}(\beta^\ell(y))\right\}_{(t,y)\in Q}
\end{align} 
is a $p$-dimensional process on $Q=[\epsilon,1]\times I$, and  define 
\begin{equation*}
    \ell^\infty(Q,\mathbb{R}^p):=\Big \{f:Q\rightarrow \mathbb{R}^p \vert\ ||f||_\infty=||(f_1,\ldots,f_p)^\top ||_\infty :=\max_{i=1}^p||f_i||_\infty <\infty\Big \}
\end{equation*}
 as the set of all $p$-dimensional functions defined on $Q$ with bounded components.
 Our next result establishes the weak convergence of the process defined by  \eqref{det101} in $ \ell^\infty(Q,\mathbb{R}^p)$.

\begin{theorem}\label{thm2}
Under Assumption \ref{ass:1}, it holds  for $\ell=1,2$
\begin{equation*}
\left\{\tfrac{1}{\sqrt{n_\ell}}S^\ell_{\lfloor n_\ell t \rfloor}(\beta^\ell(y))\right\}_{(t,y)\in Q} \leadsto \left\{\mathbb{G}^\ell(t,y)\right\}_{(t,y)\in Q}
\end{equation*}
in $\ell^\infty(Q,\mathbb{R}^p)$, where $\mathbb{G}^\ell$ is a centered Gaussian process with covariance structure
\begin{equation*}
Cov\big(\mathbb{G}^\ell(t_1,y_1),\mathbb{G}^\ell(t_2,y_2)\big)=(t_1 \land t_2)\ C^\ell(y_1,y_2)
\end{equation*}
with
\begin{eqnarray*}
C^\ell(y_1,y_2) = \mathbb{E}\bigg[X_i^\ell X_i^{\ell^\top }  \frac{\lambda\big(X_i^{\ell^\top }\beta^\ell(y_1)\big)\lambda\big(X_i^{\ell^\top }\beta^\ell(y_2)\big)\Lambda\big(X_i^{\ell^\top }\beta^\ell(y_1)\big)\Lambda\big(X_i^{\ell^\top }\beta^\ell(y_2)\big)}{\big(1-\Lambda\big(X_i^{\ell^\top }\beta^\ell(y_1)\big)\big)\big(1-\Lambda\big(X_i^{\ell^\top }\beta^\ell(y_2)\big)\big)}\cdot \\
\Big\{\Lambda\big(X_i^{\ell^\top }\beta^\ell(y_1 \land y_2)\big)-\Lambda\big(X_i^{\ell^\top }\beta^\ell(y_1)\big)\Lambda\big(X_i^{\ell^\top }\beta^\ell(y_2)\big)\Big\}
\bigg].
\end{eqnarray*}
\end{theorem}
\begin{proof} 
By Theorem 5.2 in \citet{chernozhukov:2013}, for $\ell =1,2$, the process 
$$
\left\{\tfrac{1}{\sqrt{n_\ell}}S^\ell_{\lfloor n_\ell \rfloor}(\beta^\ell(y))\right\}_{y\in I}
$$
fulfills a functional central limit theorem, where the covariance kernel of the limiting (centered) Gaussian process in $\ell^\infty (I) $ is given by $C^\ell(y_1,y_2)$. The statement of Theorem \ref{thm2} then follows with the equivalence of the conditions (A) and (F) in Theorem 1.1 in \citet{sheehy:1992}.
\end{proof}

A direct consequence of Theorem \ref{thm1} and \ref{thm2} is the following corollary.

\begin{kor}
Under Assumption \ref{ass:1}, it holds for $\ell=1,2$
\begin{eqnarray*}
\left\{\sqrt{n_\ell}\big(\hat \beta ^{\ell}(t,y)- \beta ^{\ell}(y)\big)\right\}_{(t,y)\in Q} \leadsto \left\{{1 \over t} \mathbb{G}^\ell(t,y)\right\}_{(t,y)\in Q}
\end{eqnarray*}
in $\ell^\infty(Q,\mathbb{R}^p)$, where the $p$-dimensional Gaussian process $\mathbb{G}^\ell$ is defined in Theorem \ref{thm2}.
\end{kor}
We now consider the mapping 
\begin{eqnarray}
\label{det1}
    \Phi: \begin{cases}
      \ell^\infty(Q,\mathbb{R}^p)\times\ell^\infty(Q,\mathbb{R}^p) &\rightarrow\ \ell^\infty(Q) \\
      \big(\gamma^1,\gamma^2\big) &\rightarrow\ \Phi\big(\gamma^1,\gamma^2\big)
    \end{cases}~,
\end{eqnarray}
where the function $\Phi\big(\gamma^1,\gamma^2\big): Q \to \mathbb{R}$ is defined by 
\begin{equation}
\label{det2?}
    \Phi\big(\gamma^1,\gamma^2\big)(t,y)=\Lambda\big(x^\top \gamma^1(t,y)\big) - \Lambda\big(x^\top \gamma^2(t,y)\big).
\end{equation}
In the following we interpret the function $\beta^\ell: y\mapsto\beta^\ell(y)$ as a function $(y,t)\mapsto \beta^\ell(y)$ defined on the domain $Q$, which is constant with respect to the second argument $t$. This  gives the representation 
\begin{eqnarray*}
\Phi\big(\beta^1,\beta^2\big)(t,y) &=& \Lambda\big(x^\top \beta^1(y)\big) - \Lambda\big(x^\top \beta^2(y)\big) \\
&=& F^1_{Y|X}(y|x) - F^2_{Y|X}(y|x) \\
&=:& \Delta(y|x).
\end{eqnarray*}
Interpreting the estimators $\hat \beta^\ell$ in the same way, we have 
\begin{eqnarray*}
\Phi\big(\hat\beta^1,\hat\beta^2\big)(t,y) &=& \Lambda\big(x^\top \hat\beta^1(y)\big) - \Lambda\big(x^\top \hat\beta^2(y)\big) \\
&=& \hat F^1_{Y|X}(t,y|x) - \hat F^2_{Y|X}(t,y|x) \\
&=:& \hat\Delta(t,y|x),
\end{eqnarray*}
and by the functional Delta method (\citealp{VaartWellner2023}) we obtain the following result.

\begin{theorem}\label{thm3}
    Under Assumption \ref{ass:1} and $n=n_1+n_2$, $\tfrac{n_1}{n_1+n_2}\rightarrow c$, it holds
    \begin{equation*}
   \left\{\sqrt{n}\big(\hat\Delta(t,y|x)- \Delta(y|x)\big)\right\}_{(t,y)\in Q} \leadsto \left\{\mathbb{H}(t,y)\right\}_{(t,y)\in Q} ~, 
\end{equation*}
where $\mathbb{H}$ is a centered real valued Gaussian process with covariance structure
\begin{equation} \label{det3}
    Cov\big(\mathbb{H}(t_1,y_1),\mathbb{H}(t_2,y_2)\big)= {t_1 \land t_2  \over t_1t_2} \cdot H(y_1,y_2)
\end{equation}
and
\begin{equation}\label{eq:H}
  H(y_1,y_2) = {1 \over c} \lambda\big(x^\top \beta^1(y_1)\big )^2 x^\top C^1(y_1,y_2) x + {1 \over 1-c} \lambda\big(x^\top \beta^2(y_2)\big )^2 x^\top C^2(y_1,y_2) x.
\end{equation}
\end{theorem}
\begin{proof}
Note that  the mapping $\Phi$ defined 
 in \eqref{det1} is Hadamard differentiable at the point $(\beta^1,\beta^2) \in  \ell^\infty(Q,\mathbb{R}^p)\times\ell^\infty(Q,\mathbb{R}^p) $ with derivative
    \begin{eqnarray}
    \Phi^\prime _{(\beta^1,\beta^2) }: 
    \begin{cases}
      \ell^\infty(Q,\mathbb{R}^p)\times\ell^\infty(Q,\mathbb{R}^p) &\rightarrow\ \ell^\infty(Q) \\
      \big(\gamma^1,\gamma^2\big) &\rightarrow\ \Phi 
      ^\prime _{(\beta^1,\beta^2)}
      \big(\gamma^1,\gamma^2\big)
    \end{cases}
\end{eqnarray}
where the function $\Phi 
      ^\prime _{(\beta^1,\beta^2)}
      \big(\gamma^1,\gamma^2\big): Q \to \mathbb{R} $ is defined by 
\begin{equation}
\label{det2}
    \Phi ^\prime _{(\beta^1,\beta^2)}\big(\gamma^1,\gamma^2\big)(t,y)=\Lambda^\prime \big(x^\top \beta^1(y)\big )x^\top\gamma^1(t,y) - \Lambda^\prime \big(x^\top \beta^2(y) \big )x^\top\gamma^2(t,y).
\end{equation}
Then, by the functional delta method it follows that  
\begin{align*} &\left\{\sqrt{n}\big(\hat\Delta(t,y|x)- \Delta(y|x)\big)\right\}_{(t,y)\in Q} \\ & =   \sqrt{n} \left\{
     \Phi\big(\hat\beta^1,\hat\beta^2\big)(t,y) - 
    \Phi\big(\beta^1,\beta^2\big)(t,y) 
     \right\}_{(t,y)\in Q} \\
    &=   \Phi ^\prime _{(\beta^1,\beta^2)}\Big( \sqrt{n}\big \{ \hat \beta ^{1}(t,y)- \beta ^{1}(y),
   \hat \beta ^{2}(t,y)- \beta ^{2}(y) \big \} _{(t,y)\in Q}
       \Big )  + o_{\mathbb{P}}(1) \\
    &  \leadsto \left\{\mathbb{H}(t,y)\right\}_{(t,y)\in Q}:= 
     \Phi ^\prime _{(\beta^1,\beta^2)}\Big( \Big \{ {1 \over\sqrt{c} } {1 \over\  t} \mathbb{G}^1(t,y) , {1 \over\sqrt{1-c} }{1 \over  t} \mathbb{G}^2(t,y)
      \Big \} _{(t,y)\in Q}
       \Big )~,
\end{align*}
where the $p$-dimensional Gaussian processes $\mathbb{G}^1$  and  $\mathbb{G}^2$ are defined in Theorem \ref{thm2}. Observing the definition of the derivative in \eqref{det2}, we see that
\begin{equation*}
\mathbb{H}(t,y) = \Lambda^\prime \big(x^\top \beta^1(y)\big )x^\top {1 \over\sqrt{c} } {1 \over\  t} \mathbb{G}^1(t,y) - \Lambda^\prime \big(x^\top \beta^2(y) \big )x^\top{1 \over\sqrt{1-c} } {1 \over\  t} \mathbb{G}^2(t,y).
\end{equation*}
It is now straightforward to show that the covariance structure of the limiting process factorizes as stated in \eqref{det3}. 
\end{proof}
 
\end{document}